\documentclass[11pt]{article}

\usepackage{ifthen}

\ifthenelse{\isundefined{\GenerateShortVersion}}
{
\def\LongVersion{}
\def\LongVersionEnd{}
\long\def\ShortVersion#1\ShortVersionEnd{}
}
{
\def\ShortVersion{}
\def\ShortVersionEnd{}
\long\def\LongVersion#1\LongVersionEnd{}
}

\newcommand{\Hide}[1]{\ignorespaces}

\usepackage{amsmath, amssymb}

\usepackage{amsthm}

\usepackage{ifpdf}

\ifpdf
\usepackage[pdftex]{graphicx}
\usepackage[pdftex,%
  ocgcolorlinks,%
  linkcolor=blue,%
  filecolor=blue,%
  citecolor=blue,%
  urlcolor=blue]{hyperref}
\else
\usepackage[dvips]{graphicx}
\fi

\ifthenelse{\isundefined{\NoGenerateLetterSize}}
{
\ifpdf
\setlength{\pdfpagewidth}{8.5in}
\setlength{\pdfpageheight}{11in}
\else

\fi
}
{}

\LongVersion 
\LongVersionEnd 

\ShortVersion 
\ShortVersionEnd 

\ShortVersion 
\usepackage{times}
\ShortVersionEnd 

\ShortVersion 
\renewcommand{\paragraph}[1]{\par\noindent\textbf{#1}}
\ShortVersionEnd 

\usepackage{algpseudocode}
\usepackage{algorithm}

\makeatletter
\renewcommand*{\ALG@name}{Pseudocode}
\makeatother

\usepackage{tikz}
\usetikzlibrary{arrows, positioning, calc}

\newtheorem{theorem}{Theorem}[section]
\newtheorem{lemma}[theorem]{Lemma}
\newtheorem{observation}[theorem]{Observation}
\newtheorem{corollary}[theorem]{Corollary}
\newtheorem{proposition}[theorem]{Proposition}

\theoremstyle{definition}
\newtheorem{property}[theorem]{Property}
\theoremstyle{plain}

\newtheorem*{observation*}{Observation}

\Hide{
\newtheorem{theorem}{Theorem}
\newtheorem{lemma}[theorem]{Lemma}
\newtheorem{observation}[theorem]{Observation}
\newtheorem{corollary}[theorem]{Corollary}
\newtheorem{proposition}[theorem]{Proposition}

\theoremstyle{definition}
\newtheorem{property}[theorem]{Property}
\theoremstyle{plain}
}

\Hide{
\newtheorem{theorem}{Theorem}[section]
\newtheorem{lemma}{Lemma}[section]
\newtheorem{observation}{Observation}[section]
\newtheorem{corollary}{Corollary}[section]

\theoremstyle{definition}

\theoremstyle{plain}
}

\newenvironment{DenseItemize}[0]
{\begin{list}{\labelitemi}{\usecounter{enumi}
\itemsep 0pt \parsep 0pt \leftmargin 5mm \labelwidth 5mm \parskip 0pt \topsep
0pt}}
{\end{list}}

\newenvironment{DenseEnumerate}[1][\theenumi.]
{\begin{list}{#1}{\usecounter{enumi}
\itemsep 0pt \parsep 0pt \leftmargin 5mm \labelwidth 5mm \parskip 0pt \topsep
0pt}}
{\end{list}}

\LongVersion 
\newenvironment{MathMaybe}[0]
{\begin{displaymath}\ignorespaces}
{\end{displaymath}\ignorespacesafterend}
\LongVersionEnd 
\ShortVersion 
\newenvironment{MathMaybe}[0]
{\begin{math}\ignorespaces}
{\end{math}}
\ShortVersionEnd 

\newenvironment{IntuitionSpotlight}[0]
{\par%
\setlength{\leftskip}{0.5\parindent}%
\setlength{\rightskip}{0.5\parindent}%
\noindent\itshape\textbf{Intuition spotlight:}}
{\par\ignorespacesafterend}

\newcommand{\NotationLabel}[1]{\label{notationTable:#1}\ignorespaces}
\newcommand{\NotationPageRef}[1]{\pageref{notationTable:#1}}

\newcommand{\Sect}{Sec.}
\newcommand{\Thm}{Thm.}
\newcommand{\Lem}{Lem.}
\newcommand{\Cor}{Cor.}
\newcommand{\Obs}{Obs.}
\newcommand{\Fig}{Fig.}
\newcommand{\Tabl}{Tab.}
\newcommand{\Appendix}{Apx.}

\newcommand{\Location}{\ell}
\newcommand{\aTime}{\mathit{t}}
\newcommand{\Reals}{\mathbb{R}}
\newcommand{\Integers}{\mathbb{Z}}
\newcommand{\Expect}{\mathbb{E}}
\newcommand{\Prob}{\mathbb{P}}
\newcommand{\sCost}{\mathrm{cost}^{\mathit{s}}}
\newcommand{\tCost}{\mathrm{cost}^{\mathit{t}}}
\newcommand{\Cost}{\mathrm{cost}}
\newcommand{\Leaves}{\mathcal{L}}
\newcommand{\LCA}{\mathrm{lca}}
\newcommand{\Depth}{\mathrm{depth}}
\newcommand{\Height}{\mathrm{height}}
\newcommand{\Ancestors}{\mathit{anc}}
\newcommand{\Indicator}{\mathbf{1}}
\newcommand{\ExpDist}{\mathrm{Exp}}
\newcommand{\PoisDist}{\mathrm{Pois}}
\newcommand{\EndTime}{\mathit{t}_{\mathrm{end}}}
\newcommand{\sEndCost}{\mathit{c}_{\mathrm{end}}^{\mathit{s}}}
\newcommand{\LateTime}{\mathit{t}_{\mathrm{late}}}

\newcommand{\Active}{\mathit{C}}
\newcommand{\Odd}{\mathit{D}}
\newcommand{\Effective}{\mathit{F}}
\newcommand{\Stilts}{\mathcal{S}}
\newcommand{\Heads}{\mathit{H}}

\newcommand{\A}{\mathcal{A}}
\newcommand{\adv}[1]{{#1}^{*}}
\newcommand*{\xor}{\mathbin{\oplus}}
\begin{document}

\LongVersion 
\title{Online Matching: Haste makes Waste! \\ (Full Version)\footnote{%
An extended abstract will appear in Proceedings of ACM STOC 2016.
}}
\LongVersionEnd 
\ShortVersion 
\title{Online Matching: Haste makes Waste! \\ (Extended Abstract)}
\ShortVersionEnd 

\date{}

\author{%
Yuval Emek\thanks{%
Technion, Israel.
Email: \texttt{yemek@ie.technion.ac.il}.
Partially supported by the Technion-Microsoft Electronic Commerce Research
Center.}
\and
Shay Kutten\thanks{%
Technion, Israel.
Email: \texttt{kutten@ie.technion.ac.il}.
Partially supported by the Technion-Microsoft Electronic Commerce Research
Center and by the France Israel cooperation grant from the Israeli Ministry of
Science.}
\and
Roger Wattenhofer\thanks{%
ETH Zurich, Switzerland.
Email: \texttt{wattenhofer@ethz.ch}.}
}

\begin{titlepage}

\maketitle

\begin{abstract}
This paper studies a new online problem, referred to as \emph{min-cost
perfect matching with delays (MPMD)}, defined over a finite metric
space (i.e., a complete graph with positive edge weights obeying
the triangle inequality) $\mathcal{M}$ that is known to the algorithm in
advance.
Requests arrive in a continuous time online fashion at the points of
$\mathcal{M}$ and should be served by matching them to each other.
The algorithm is allowed to delay its request matching commitments, but this
does not come for free:
the total cost of the algorithm is the sum of metric distances between matched
requests \emph{plus} the sum of times each request waited since it arrived
until it was matched.
A randomized online MPMD algorithm is presented whose competitive
ratio is
$O (\log^{2} n + \log \Delta)$,
where $n$ is the number of points in $\mathcal{M}$ and $\Delta$ is its aspect
ratio.
The analysis is based on a machinery developed in the context of a new
stochastic process that can be viewed as two interleaved Poisson processes;
surprisingly, this new process captures precisely the behavior of our
algorithm.
A related problem in which the algorithm is allowed to clear any unmatched
request at a fixed penalty is also addressed.
It is suggested that the MPMD problem is merely the tip of the iceberg for a
general framework of online problems with delayed service that captures
many more natural problems.
\end{abstract}

\renewcommand{\thepage}{}
\end{titlepage}

\pagenumbering{arabic}

\section{Introduction}
\label{section:introduction}
Consider an online gaming platform supporting two-player games such as Chess,
Scrabble, or Street Fighter 4.\footnote{%
Leading gaming platforms include XBOX Live and the Playstation Network for
consoles, Steam for PCs, and web-based platforms such as Geewa, Pogo and Yahoo
Games.}
The platform tries to find a suitable opponent for each player connecting to
it;
matching two players initiates a new game between them.
The platform should minimize two criteria:
(i) the difference between the matched players' rating (a positive integer
that represents the player's skill), so that the game is challenging for both
players; and
(ii) the waiting time until a player is matched and can start playing since
waiting is boring.
(In reality, the $1$-dimensional player rating space is often generalized to a
more complex metric space by taking into account additional parameters such as
the network distance between the matched players.)
It turns out, though, that these two minimization criteria are often
conflicting:
What if the pool of players waiting for a suitable opponent does not contain
anyone whose rating is close to that of a new player?
Should the system match the new player to an opponent whose rating differs
significantly from hers?

The naive approach that matches players immediately does
a terrible job:
Murphy's Law may strike, and right after matching a player, a perfect opponent
will emerge:
Haste makes waste, unbounded waste in fact.
To cope with this challenge, we must allow the platform to delay its
service in a \emph{rent-or-buy} manner.

\paragraph{Model.}
Let
$\mathcal{M} = (V, \delta)$
be a finite metric space.
Consider a set $R$ of \emph{requests}, where each request $\rho \in R$
is characterized by its \emph{location}
$\Location(\rho) \in V$
\NotationLabel{model:location}
(also referred to as the point that \emph{hosts} $\rho$) and \emph{arrival
time}
$\aTime(\rho) \in \Reals_{\geq 0}$.\footnote{
For ease of reference, \Tabl{}~\ref{table:notation} provides an index for the
notation used throughout this paper.
}
\NotationLabel{model:arrival-time}
Assume for the time being that $|R|$ is even.\footnote{%
The problem presented here is not well defined if $|R|$ is odd however, later
on we discuss a variant of this problem that is well defined for any finite
$|R|$.
}
Notice that $R$ can have multiple requests with the same location (in
particular, $|R|$ is unbounded with respect to $|V|$);
for simplicity, we assume that each request has a unique arrival
time.\footnote{%
This assumption is without loss of generality as the arrival times can be
slightly perturbed.
}

The input to an online algorithm for the \emph{min-cost perfect matching with
delays (MPMD)} problem is a finite metric space $\mathcal{M}$,
provided to the algorithm before the execution commences, and a request set
$R$ over $\mathcal{M}$ such that each request $\rho \in R$ is presented to
the algorithm in an online fashion at its arrival time
$\aTime(\rho)$.
The goal of the algorithm is to construct a (perfect) \emph{matching} of the
request set --- namely, a partition of $R$ into $|R| / 2$ unordered request
pairs --- in an online fashion with no preemption.

The algorithm is allowed to delay the matching of any request in $R$ at a
cost.
More formally, the requests $\rho_{1}$ and $\rho_{2}$ can be matched at any time
$t \geq \max \{ \aTime(\rho_{1}), \aTime(\rho_{2}) \}$;
if algorithm $\A$ matches requests $\rho_{1}$ and $\rho_{2}$ at time $t$,
then it incurs a \emph{time cost} of
$\tCost_{\A}(\rho_{i}) = t - \aTime(\rho_{i})$
\NotationLabel{model:time-cost-request}
and a \emph{space cost} of
$\sCost_{\A}(\rho_{i}) = \delta(\Location(\rho_{1}), \Location(\rho_{2})) / 2$
\NotationLabel{model:space-cost-request}
for serving $\rho_{i}$,
$i \in \{ 1, 2 \}$.
The \emph{space cost} and \emph{time cost} of $\A$ for the whole request set
are
$\sCost_{\A}(R, \mathcal{M})
=
\sum_{\rho \in R} \sCost_{\A}(\rho)$
\NotationLabel{model:space-cost-instance}
and
$\tCost_{\A}(R, \mathcal{M})
=
\sum_{\rho \in R} \tCost_{\A}(\rho)$,
\NotationLabel{model:time-cost-instance}
respectively.
The objective is to minimize
$\Cost_{\A}(R, \mathcal{M}) = \sCost_{\A}(R, \mathcal{M}) + \tCost_{\A}(R,
\mathcal{M})$.
\NotationLabel{model:total-cost}
When $\A$ is clear from the context, we may drop the subscript.

Following the common practice in online computation (cf.\ \cite{BorodinE1998}),
the quality of an online MPMD algorithm is measured in terms of its
\emph{competitive ratio}.
Online MPMD algorithm $\A$ is said to be \emph{$\alpha$-competitive} if for
every finite metric space $\mathcal{M}$, there exists some
$\beta = \beta(\mathcal{M})$
such that for every (even size) request set $R$ over $\mathcal{M}$, it
is guaranteed that
$\Expect[\Cost_{\A}(R, \mathcal{M})]
\leq
\alpha \cdot \Cost_{\adv{\A}}(R, \mathcal{M}) + \beta$,
where the expectation is taken over the coin tosses of the algorithm
(if any) and $\adv{\A}$ is an optimal offline algorithm.
It is assumed that $\mathcal{M}$ and $R$ are generated by an \emph{oblivious
adversary} that knows $\A$, but not the realization of its coin tosses.

\paragraph{Related work.}
The \emph{rent-or-buy} feature is fundamental to
\LongVersion 
many online applications and thus, prominent in the theoretical study of
\LongVersionEnd 
online computation.
Classic online
\LongVersion 
problems in which the rent-or-buy feature constitutes the
sole source of difficulty
\LongVersionEnd 
\ShortVersion 
rent-or-buy problems
\ShortVersionEnd 
include
ski-rental
\cite{KarlinMRS1986,KarlinMMLO1990,KarlinKR2001} and
TCP acknowledgment
\cite{DoolyGS1998,DoolyGS2001,KarlinKR2001}.
\LongVersion 
In other problems, the rent-or-buy feature is
\LongVersionEnd 
\ShortVersion 
This feature is sometimes
\ShortVersionEnd 
combined with a complex combinatorial structure, enhancing an already
challenging online problem, e.g., the extension of online job scheduling
\cite{AwerbuchKP1992,AwerbuchALR2002} studied in
\cite{AverbakhX2007,AverbakhB2013,AzarEJV2016}.

The \emph{matching} problem is a combinatorial optimization celebrity ever
since the seminal work of Edmonds~\cite{Edmonds1965a,Edmonds1965b}.
The realm of \emph{online algorithms} also features an extensive literature on
matching and some generalizations thereof.
Online problems that have been studied in this regard include
maximum cardinality matching
\cite{KarpVV1990,BirnbaumM2008,GoelM2008,DevanurJK2013,Miyazaki2014,NaorW2015},
maximum vertex-weighted matching
\cite{AggarwalGKM2011,DevanurJK2013,NaorW2015},
maximum capacitated assignment (a.k.a.\ the AdWords problem)
\cite{MehtaSVV2005,BuchbinderJN2007,GoelM2008,AggarwalGKM2011,NaorW2015},
metric maximum weight matching
\cite{KalyanasundaramP1993,KhullerMV1994},
metric minimum cost perfect matching
\cite{KalyanasundaramP1993,MeyersonNP2006,BansalBGN2014}, and
metric minimum capacitated assignment (a.k.a.\ the transportation
problem) \cite{KalyanasundaramP2000};
see \cite{Mehta2013} for a comprehensive survey.
All these online problems are \emph{bipartite} matching versions, where the
nodes in one side of the graph are static and the nodes in the other side are
revealed in an online fashion together with their incident edges.

\paragraph{Discussion and results.}
To the best of our knowledge,
\LongVersion 
the MPMD problem
\LongVersionEnd 
\ShortVersion 
MPMD
\ShortVersionEnd 
is the first online
\emph{all-pairs} matching version.
Moreover, in contrast to the previously studied online matching versions, in
MPMD the graph (or metric space) is known a-priori and the algorithmic
challenge stems from the unknown locations and arrival times of the
\LongVersion 
requests (whose number is unbounded);
\LongVersionEnd 
\ShortVersion 
requests;
\ShortVersionEnd 
this is more in the spirit of online problems such as the classic $k$-server
problem.

\ShortVersion \sloppy \ShortVersionEnd
The main
\LongVersion technical \LongVersionEnd
result of this paper is a randomized online MPMD algorithm
whose competitive ratio is
$O (\log^{2} n + \log \Delta)$,
where $n$ is the number of points in the metric space $\mathcal{M}$ and
$\Delta$ is its aspect ratio.
This algorithm, presented in \Sect{}~\ref{section:algorithm}, is based on
\emph{exponential timers} that determine how long should we wait before
committing to a certain match.
The analysis, presented in \Sect{}~\ref{section:analysis}, relies heavily on
machinery we develop in the context of a new stochastic process named
\emph{alternating Poisson process}.
\ShortVersion \par\fussy \ShortVersionEnd

We also consider a variant of the online MPMD problem, referred to as
\emph{MPMDfp}, in which the algorithm can clear any unmatched request at a
fixed penalty.
This problem variant is motivated by noticing that clearing an unmatched
request may correspond to matching a player with a computer opponent in the
context of the aforementioned gaming platforms.
The MPMDfp problem is discussed further in
\LongVersion 
\Sect{}~\ref{section:fixed-penalty},
\LongVersionEnd 
\ShortVersion 
the full version,
\ShortVersionEnd 
where we show that our online algorithm can be adjusted to cope with this
variant as well.

It is not difficult to develop constant lower bounds on the competitive ratio
of online MPMD algorithms already for the special case of a $2$-point metric
\LongVersion 
space (note that this special case generalizes the ski rental problem).
\LongVersionEnd 
\ShortVersion 
space.
\ShortVersionEnd 
While a $2$-point metric space admits an $O (1)$-competitive online MPMD
algorithm, we conjecture that in the general case, the competitive ratio must
grow as a function of $n$.
In particular, we believe that this conjecture holds for the $1$-dimensional
metric spaces constructed in Appendix~C of \cite{EmekLW2015} (a variant of the
construction in Fig.~1 of \cite{ReingoldT1981}).
We also establish an algorithm-specific lower bound:
To demonstrate the role of randomness in the online algorithm presented in
\Sect{}~\ref{section:algorithm}, we show in
\LongVersion 
\Sect{}~\ref{section:specific-lower-bound}
\LongVersionEnd 
\ShortVersion 
the full version
\ShortVersionEnd 
that the competitive ratio of its natural deterministic counterpart is $\Omega
(n)$.

\paragraph{Online problems with delayed service.}
The online MPMD problem is obtained by augmenting the (offline) min-cost
perfect matching problem with the time axis over which service can be delayed
in a rent-or-buy manner.
This viewpoint seems to open a gate to a general framework of online problems
with delayed service since the approach of combining the rent-or-buy feature
with a combinatorial optimization offline problem can be applied to a class of
minimization problems much larger than just min-cost perfect matching.

To be more precise, consider a minimization problem $\mathcal{P}$ defined with
respect to some underlying combinatorial structure $\mathcal{C}$ with a ground
set
$\mathcal{E}_{\text{in}}$
of \emph{input entities} and a ground set
$\mathcal{E}_{\text{out}}$
of \emph{output entities}.
The input and output instances of $\mathcal{P}$ are multisets over
$\mathcal{E}_{\text{in}}$
and
$\mathcal{E}_{\text{out}}$,
respectively.
For each input instance $I$, problem $\mathcal{P}$
determines a collection $\mathcal{F}(I)$ of \emph{feasible}
output instances;
input instance $I$ is said to be \emph{admissible} if
$|\mathcal{F}(I)| \neq \emptyset$.
We restrict our attention to problems $\mathcal{P}$ satisfying
the property that for every two input instances
$I \subseteq J$,
if $I$ and $J$ are admissible, then so is
$J - I$.\footnote{%
We follow the standard multiset convention that for two multisets $M, N$ over
a ground set $S$ with multiplicity functions
$\mu_{M} : S \rightarrow \Integers_{\geq 0}$
and
$\mu_{N} : S \rightarrow \Integers_{\geq 0}$,
the relation
$M \subseteq N$
holds if
$\mu_{M}(x) \leq \mu_{N}(x)$
for every
$x \in S$; and
$N - M$
is the multiset whose multiplicity function
$\mu_{N - M} : S \rightarrow \Integers_{\geq 0}$
satisfies
$\mu_{N - M}(x) = \max\{ \mu_{N}(x) - \mu_{M}(x), 0 \}$
for every
$x \in S$.
}

Minimization problem $\Pi$ can be transformed into an online
problem with delayed service $\Pi_{\text{on}}$ by applying to it the
\emph{delayed service operator}:
Each request in $\Pi_{\text{on}}$ is characterized by its \emph{location} ---
an entity in $\mathcal{E}_{\text{in}}$ ---  and by its arrival time.
The algorithm can serve a collection $R$ of yet unserved requests by buying a
feasible (under $\mathcal{F}$) output instance $S$ for their location multiset
at any time $t$ after the arrival of all requests in $R$.
The payment for this service is the cost of $S$ plus the total waiting times
of the requests in $R$ up to time $t$.
Notice that this act of buying $S$ does not serve requests other than those in
$R$ including any request arriving at the locations of $R$ after time $t$.

\LongVersion 
The online MPMD problem
\LongVersionEnd 
\ShortVersion 
MPMD
\ShortVersionEnd 
is obtained by applying this delayed service operator to the
metric min-cost perfect matching problem, where
$\mathcal{C}$ is a finite metric space,
$\mathcal{C}_{\text{in}}$ is its points, and
$\mathcal{C}_{\text{out}}$ is the set of
\LongVersion 
unordered
\LongVersionEnd 
point pairs (a point multiset is an admissible input instance if its
cardinality is even).\footnote{%
In the offline version of the metric min-cost perfect matching problem it
suffices to consider only sets (rather than multisets) for the input and
output instances.
The generalization to multisets is necessary for the transition to the online
version of the problem.
}
This operator can also be applied to
the vertex cover problem
($\mathcal{C}$ is a graph, $\mathcal{C}_{\text{in}}$ is the edge set, and
$\mathcal{C}_{\text{out}}$ is the vertex set),
the dominating set problem
($\mathcal{C}$ is a graph and $\mathcal{C}_{\text{in}}$ and
$\mathcal{C}_{\text{out}}$ are the vertex set),
and many more combinatorial optimization problems.

\section{Preliminaries}
\label{section:preliminaries}

\paragraph{Tree notation and terminology.}
Consider a tree $T$ \emph{rooted} at some vertex $r$ with a \emph{leaf} set
$\Leaves$.
The notions \emph{parent}, \emph{ancestor}, \emph{child}, and \emph{sibling}
are used in their usual sense.
A \emph{binary} tree is called \emph{full} if every internal vertex has
exactly two children.

Let $v$ be some vertex in $T$.
The parent of $v$ in $T$ (assuming that $v \neq r$) is denoted by
$p(v)$.
\NotationLabel{tree:parent}
We denote the \emph{subtree} of $T$ rooted at $v$ by $T(v)$
\NotationLabel{tree:subtree}
and the leaf set of $T(v)$ by $\Leaves(v)$.
\NotationLabel{tree:subtree-leaves}
The set of ancestors of $v$ (excluding $v$ itself) is denoted by
$\Ancestors(v)$.
\NotationLabel{tree:ancestors}
The \emph{depth} of $v$ in $T$ --- i.e., the distance (in hops) from $v$ to
$r$ --- is denoted by
$\Depth(v)$
\NotationLabel{tree:depth}
and the \emph{height} of $T$ is denoted by
$\Height(T) = \max_{x \in \Leaves} \Depth(x)$.
\NotationLabel{tree:height}

A \emph{stilt} in $T$ is an oriented path connecting some vertex $v \in T$,
referred to as the \emph{head} of the stilt, with a leaf in $\Leaves(v)$,
referred to as the \emph{foot} of the stilt.
Given two leaves $x, y \in \Leaves$, their \emph{least common ancestor (LCA)}
in $T$ is denoted by
$\LCA(x, y)$.
\NotationLabel{tree-lca}

\paragraph{Probabilistic embedding in tree metric spaces.}
Let
$w : T \rightarrow \Reals_{\geq 0}$
be a weight function on the vertices of $T$ that satisfies
(i) $w(v) = 0$ for every leaf $v \in \Leaves$; and
(ii) $w(v) < w(p(v))$
for every vertex
$v \in T - \{ r \}$.
The pair
$(T, w)$
introduces a finite metric (in fact, an ultrametric) space over the leaf set
$\Leaves$ with distance function $\delta$ defined by setting
$\delta(x, y) = w(\LCA(x, y))$
for every $x, y \in \Leaves$.
A metric space that can be realized by such a $(T, w)$ pair is referred to as
a \emph{tree metric space}.
We subsequently identify a tree metric space with the pair $(T, w)$ that
realizes it.

Consider some real $\alpha > 1$.
A \emph{hierarchically well separated tree} with parameter $\alpha$
(cf.\ \cite{Bartal1996}), or \emph{$\alpha$-HST} in short, is a tree metric
space $(T, w)$ that, in addition to the aforementioned
requirements, satisfies
$w(p(v)) \geq \alpha \cdot w(v)$
for every vertex
$v \in T - \{ r \}$.
We refer to an $\alpha$-HST realized by a full binary tree $T$ (cf.\
\cite{CoteMP2008}) as an \emph{$\alpha$-HSBT}.

The following theorem is established by combining a celebrated construction of
Fakcharoenphol et al.~\cite{FakcharoenpholRT2004} (improving previous
constructions of Bartal \cite{Bartal1996,Bartal1998}) with a tree
transformation technique \cite{PattShamir2015}
\LongVersion 
(details are deferred to \Appendix{}~\ref{appendix:embedding-in-HSBT}).
\LongVersionEnd 
\ShortVersion 
(details are deferred to the full version).
\ShortVersionEnd 

\begin{theorem} \label{theorem:HSBT}
Consider some $n$-point metric space
$(V, \delta)$
of aspect ratio
$\Delta = \frac{\max_{x \neq y \in V} \delta(x, y)}{\min_{x \neq y \in V}
\delta(x, y)}$
and let $\mathcal{U}$ be the set of all $(1 + \Omega (1 / \log n))$-HSBTs
$(T, w)$
over $V$ with
$\Height(T) = O (\log \Delta + \log n)$
and with
distance functions $\delta_{\mathcal{T}}$ that dominate $\delta$ in the sense that
$\delta_{\mathcal{T}}(x, y) \geq \delta(x, y)$
for every $x, y \in V$.
There exists a probability distribution $\mathcal{P}$ over $\mathcal{U}$ such
that
$\Expect_{(V, \delta_{\mathcal{T}}) \in \mathcal{P}}[\delta_{\mathcal{T}}(x, y)]
\leq
O (\log n) \cdot \delta(x, y)$
for every $x, y \in V$.
Moreover, the probability distribution $\mathcal{P}$ can be sampled
efficiently.
\end{theorem}

\paragraph{Matching algorithm notation and terminology.}
Consider the operation of an MPMD algorithm on some HSBT $(T, w)$.
\LongVersion 
Recall that the input to the algorithm consists of a finite set $R$ of
requests, where each request
$\rho \in R$
is characterized by its location
$\Location(\rho) \in \Leaves$
and arrival time
$\aTime(\rho) \in \Reals_{\geq 0}$.
\LongVersionEnd 
Suppose that the algorithm matches requests $\rho$ and $\rho'$ with
$\Location(\rho) = x \in \Leaves$
and
$\Location(\rho') = x' \in \Leaves$,
$x \neq x'$.
Let $v$ be some vertex in the unique path connecting $x$ and $x'$ in $T$.
If
$v = \LCA(x, x')$,
then we refer to this matching operation as matching \emph{across} $v$;
otherwise, we refer to it as matching \emph{on top of} $v$.
Notice that matching across $v$ corresponds to matching a request located in
$\Leaves(u_{1})$ with a request located in $\Leaves(u_{2})$, where $u_{1}$ and
$u_{2}$ are the children of $v$ in $T$, whereas matching on top of $v$
corresponds to matching a request located in $\Leaves(v)$ with a request
located in $\Leaves - \Leaves(v)$.

If the algorithm matches request
\LongVersion 
$\rho \in R$
\LongVersionEnd 
\ShortVersion 
$\rho$
\ShortVersionEnd 
at time $t'$, then $\rho$ is said to be \emph{active} at all times
$\aTime(\rho) \leq t < t'$.
Given some vertex $v \in T$, we denote the set of active requests in
$\Leaves(v)$ at time $t$ by
$\Active_{v}(t)$
\NotationLabel{alg:active}
and write
$\Active(t) = \Active_{r}(t)$.
\NotationLabel{alg:active-root}
Vertex $v$ is said to be \emph{odd} at time $t$ if
$|\Active_{v}(t)| = 1 \pmod{2}$;
let
$\Odd(t)$
\NotationLabel{alg:odd}
be the set of odd vertices at time $t$.

A key observation is that the forest induced on $T$ by the vertex subset
$\Odd(t)$ is a collection --- denoted hereafter by $\Stilts(t)$ ---
\NotationLabel{alg:stilts}
of vertex disjoint stilts.
Moreover, if $v$ is the head of a stilt in $\Stilts(t)$ then either
(1) $v = r$ is the root of $T$ (which implies that $|\Active(t)|$ is odd); or
(2) the sibling of $v$ is also the head of a stilt in $\Stilts(t)$.
Let
$\Heads(t) \subseteq \Odd(t)$
\NotationLabel{alg:heads}
be the set of heads of stilts in $\Stilts(t)$.

Internal vertex $v \in T - \Leaves$ is said to be \emph{effective} at time $t$
if its two children are odd (which, in particular, means that $v$ is not odd);
let
$\Effective(t)$
\NotationLabel{alg:effective}
be the set of effective vertices at time $t$.
Notice that $v$ is effective if and only if its two children are in
$\Heads(t)$ and let
$S_{1}, S_{2} \in \Stilts(t)$
be their corresponding stilts.
We refer to the feet of $S_{1}$ and $S_{2}$ as the
\emph{supporting leaves} of $v$ at time $t$.

We shall apply the aforementioned matching algorithm definitions to both our
online MPMD algorithm, denoted by $\A$, and to the benchmark offline MPMD
algorithm, denoted by $\adv{\A}$.
To distinguish between the two, we reserve the
aforementioned notation system for the former and add a superscript asterisk
for the latter;
in particular, the set of vertices odd under $\adv{\A}$ at time $t$ is denoted
by
$\adv{\Odd}(t)$
\NotationLabel{alg:adv-odd}
(whereas the set of vertices odd under $\A$ at time $t$ is denoted by
$\Odd(t)$).

\section{An online MPMD algorithm}
\label{section:algorithm}
In this section, we present our online MPMD algorithm, referred to as the
\emph{stilt-walker algorithm} and denoted hereafter by $\A$;
its competitive ratio is analyzed in \Sect{}~\ref{section:analysis}.
The algorithm works in two stages:
a preprocessing stage, in which we employ \Thm{}~\ref{theorem:HSBT} to embed
the input metric space in a random $(1 + \Omega (1 / \log n))$-HSBT $(T, w)$,
and the actual online execution, in which $\A$ processes the requests arriving
at the leaves of $T$ and constructs the desired matching.
The remainder of this section is dedicated to describing the latter.

\paragraph{The matching policy.}
Although $\A$ operates in continuous time, it will be convenient to describe
it as if it progresses in discrete \emph{time steps}, taking the difference $d
t$ between two consecutive time steps to be infinitesimally small so that at
most one request arrives in each time step.

Fix some time step $t$.
If request $\rho$ arrives at this time step and $\Location(\rho)$ already
hosts another active (under $\A$) request $\rho'$, then the algorithm matches
$\rho$ and $\rho'$ immediately.
Assume hereafter that each leaf in $\Leaves$ hosts at most one active
request.

Consider some effective vertex $v \in \Effective(t)$ and let $x_{1}^{v},
x_{2}^{v}$ be its supporting leaves (the feet of the corresponding stilts in
$\Stilts(t)$).
By definition, $x_{i}^{v}$ hosts an odd number of active requests at time $t$ for
$i \in \{ 1, 2 \}$ and since it cannot host more than one active request, it
follows that there exists a unique active request $\rho_{i}^{v}$ at time $t$
with
$\Location(\rho_{i}^{v}) = x_{i}^{v}$;
we refer to $\rho_{1}^{v}$ and $\rho_{2}^{v}$ as the
\emph{supporting requests} of $v$.
The algorithm tosses an independent biased coin and matches its supporting
requests (i.e., matching across $v$) with probability
$d t / w(v)$.
In what follows, we attribute this coin toss to $v$ so that we can distinguish
between coin tosses of different (internal) vertices.
A pseudocode description of the stilt-walker algorithm is provided in
Pseudocode~\ref{algorithm:stilt-walker}.

\def\PseudocodeStiltWalker{
\begin{algorithm}
\caption{\label{algorithm:stilt-walker}%
The operation of $\A$ at time step $t$.
}
\begin{algorithmic}[1]
\If{$\exists \rho, \rho' \in \Active(t)$ with $\Location(\rho) =
\Location(\rho')$} \Comment{there can be at most one such request pair}
  \State{match $\rho$ and $\rho'$}
\EndIf
\ForAll{$v \in \Effective(t)$}
  \State{$x_{1}^{v}, x_{2}^{v} \leftarrow$ supporting leaves of $v$}
  \State{$\rho_{i}^{v} \leftarrow$ unique active request with
$\Location(\rho_{i}^{v}) = x_{i}^{v}$ for $i = 1, 2$}
  \State{$z = z(v, t) \leftarrow$ outcome of an independent Bernoulli trial
with parameter $d t / w(v)$}
  \If{$z = 1$}
    \State{match $\rho_{1}^{v}$ and $\rho_{2}^{v}$} \Comment{matching across $v$}
  \EndIf
\EndFor
\end{algorithmic}
\end{algorithm}
}

\LongVersion 
\PseudocodeStiltWalker{}
\LongVersionEnd 

An analogous ``continuous'' description of the stilt-walker algorithm's policy
regarding the effective vertices is as follows.
Consider some internal vertex $v \in T - \Leaves$ and suppose that the last
time $\A$ matched across $v$ was at time $t_{0}$ (take
$t_{0} = 0$
if $\A$ still has not matched across $v$).
Then, the next time the algorithm matches across $v$ is the minimum $t_{1}$
that satisfies
\begin{MathMaybe}
\int_{t_{0}}^{t_{1}} \Indicator(v \in \Effective(t)) \, d t
=
Z \, ,
\end{MathMaybe}
where
$\Indicator(\cdot)$
denotes the indicator operator and
$Z = Z(v, t_{0}) \sim \ExpDist(1 / w(v))$
is an (independent) random variable that obeys an exponential distribution
with rate $1 / w(v)$.

\begin{IntuitionSpotlight}
The reader may wonder about the role of the exponential timers maintained at
the internal vertices.
At first, we tried to analyze the deterministic version of the algorithm,
where the $(1 / w(v))$-rate exponential timer maintained at vertex $v \in T -
\Leaves$ is replaced by a deterministic $\Theta (w(v))$-timer.
This seemed to make sense because it allows the algorithm to wait for $\Theta
(w(v))$ time before it pays $w(v)$ in space cost (the usual approach to
rent-or-buy problems).
However, as demonstrated in
\LongVersion 
\Sect{}~\ref{section:specific-lower-bound},
\LongVersionEnd 
\ShortVersion 
the full version,
\ShortVersionEnd 
this is hopeless.
Switching to the randomized version resolves this obstacle because the
memoryless exponential timers allow us to analyze each vertex independently
and partition the time into periods so that each period can be analyzed
independently --- see \Sect{}~\ref{section:analysis-heart}.
\end{IntuitionSpotlight}

Notice that our algorithm is guaranteed to eventually match all requests with
probability $1$.
Indeed, if there are at least two active requests at time $t$, then there is
at least one effective vertex $v$ at time $t$ and $\A$ matches across it (thus
matching its supporting requests) at time $t + d t$ with probability $d t /
w(v)$.

\section{Analyzing the stilt-walker algorithm}
\label{section:analysis}
Our main goal in this section is to establish the following Theorem.

\begin{theorem} \label{theorem:main}
Fix some $1 < \alpha \leq 2$ and consider an $\alpha$-HSBT $\mathcal{T}$
realized by a full binary tree of height $h$.
Let $R$ be a request set over $\mathcal{T}$ and let $\adv{\A}$ be
some benchmark offline MPMD algorithm for $\mathcal{T}$, $R$.
The stilt-walker algorithm $\A$ guarantees that
\[
\Expect \left[ \Cost_{\A} \left( R, \mathcal{T} \right) \right]
\leq
O (1 / (\alpha - 1)) \cdot \sCost_{\adv{\A}} \left( R,
\mathcal{T} \right)
+
O (h) \cdot \tCost_{\adv{\A}} \left( R, \mathcal{T} \right)
+
\beta  \, ,
\]
where
$\beta = \beta(\mathcal{T})$
depends only on $\mathcal{T}$ and is independent of $R$.
\end{theorem}

\ShortVersion 
Showing that combining \Thm{}\ \ref{theorem:HSBT} and \ref{theorem:main}
yields the desired upper bound on the competitive ratio of $\A$ is based on
relatively standard arguments and is deferred to the full version.
\ShortVersionEnd 
\LongVersion 
We will soon turn our attention to the proof of \Thm{}~\ref{theorem:main},
but first, let us show that it yields the desired upper bound on the
competitive ratio of $\A$.
To that end, fix some $n$-point metric space
$\mathcal{M} = (V, \delta)$
of aspect ratio $\Delta$ and a request set $R$ over $\mathcal{M}$ and let
$\widetilde{\A}^{*}$ be an optimal (offline) algorithm for $R$ (over
$\mathcal{M}$).
Let $\mathcal{P}$ be the probability distribution promised by
\Thm{}~\ref{theorem:HSBT} when applied to $\mathcal{M}$.
Denoting the coin tosses of $\A$ by $\chi$ and taking $\mathcal{T}$ to be some
HSBT in the support of $\mathcal{P}$, we can employ \Thm{}~\ref{theorem:main}
to conclude that
\[
\Expect_{\chi} \left[ \Cost_{\A} \left( R,
\mathcal{T} \right) \right]
\leq
O (\log n) \cdot \sCost_{\A^{*}} \left( R, \mathcal{T} \right)
+
O (\log \Delta + \log n) \cdot \tCost_{\A^{*}} \left( R, \mathcal{T}
\right)
+
\beta(\mathcal{T}) \, ,
\]
where $\A^{*}$ is the projection of $\widetilde{\A}^{*}$ on $\mathcal{T}$
(that is, same requests are matched at the same time, incurring possibly
different space costs).
Therefore,
\begin{align*}
\Expect_{\mathcal{P}, \chi} \left[ \Cost_{\A} \left( R, \mathcal{M} \right)
\right]
\, \leq \, &
\Expect_{\mathcal{T} \in \mathcal{P}} \left[
\Expect_{\chi} \left[ \Cost_{\A} \left( R, \mathcal{T} \right) \right]
\right] \\
\leq \, &
\Expect_{\mathcal{T} \in \mathcal{P}} \left[
O (\log n) \cdot \sCost_{\A^{*}} \left( R, \mathcal{T} \right)
+
O (\log \Delta + \log n) \cdot \tCost_{\A^{*}} \left( R, \mathcal{T}
\right)
+
\beta(\mathcal{T})
\right] \\
= \, &
O (\log n) \cdot \Expect_{\mathcal{T} \in \mathcal{P}} \left[ \sCost_{\A^{*}} \left(
R, \mathcal{T} \right) \right]
+
O (\log \Delta + \log n) \cdot \tCost_{\widetilde{\A}^{*}} \left( R, \mathcal{M}
\right)
+
\beta(\mathcal{M}) \\
\leq \, &
O \left( \log^{2} n \right) \cdot \sCost_{\widetilde{\A}^{*}} \left( R,
\mathcal{M} \right)
+
O (\log \Delta + \log n) \cdot \tCost_{\widetilde{\A}^{*}} \left( R, \mathcal{M}
\right)
+
\beta(\mathcal{M}) \\
\leq \, &
O \left( \log \Delta + \log^{2} n \right) \cdot \Cost_{\widetilde{\A}^{*}} \left(
R, \mathcal{M} \right)
+
\beta(\mathcal{M}) \, ,
\end{align*}
where
$\beta(\mathcal{M}) = \Expect_{\mathcal{T} \in \mathcal{P}}[\beta(\mathcal{T})]$,
the first transition holds since the distance functions in the
support of $\mathcal{P}$ dominate $\delta$,
the third transition holds since the time costs of
$\widetilde{\A}^{*}$ in $\mathcal{M}$ are the same as those of $\A^{*}$ in
$\mathcal{T}$, and
the fourth transition holds by \Thm{}~\ref{theorem:HSBT}.
\par
\LongVersionEnd 
The remainder of this section is dedicated to the proof of
\Thm{}~\ref{theorem:main} and is organized as follows:
First, in \Sect{}~\ref{section:alternating-Poisson-process}, we introduce a
new stochastic process, called \emph{alternating Poisson process (APP)},
together with some related machinery.
APPs play a major role in \Sect{}~\ref{section:analysis-heart} that forms the
heart of the analysis:
we prove \Thm{}~\ref{theorem:main} assuming that online algorithm $\A$
receives a special \emph{end-of-input} signal upon receiving the last request
in $R$ and responds to it by immediately matching all remaining active
requests.
\LongVersion 
Finally, in \Sect{}~\ref{section:lift-end-of-input-assumption}, we lift the
assumption of receiving the end-of-input signal, showing that it does not
affect the (multiplicative) competitive ratio.
\LongVersionEnd 
\ShortVersion 
Lifting the end-of-input signal assumption is relatively straightforward and
is deferred to the full version.
\ShortVersionEnd 

\subsection{Alternating Poisson processes}
\label{section:alternating-Poisson-process}
A major component of the analysis presented in
\Sect{}~\ref{section:analysis-heart} is a stochastic process (more
specifically, a point process) that we refer to as an \emph{alternating
Poisson process (APP)}.
This process is parametrized by its
\emph{start time} $t_{0} \in \Reals_{\geq 0}$,
\emph{length} $\gamma \in \Reals_{> 0}$,
\emph{rate} $\lambda \in \Reals_{> 0}$, and
a right-continuous \emph{coloring function}
$c : [t_{0}, t_{0} + \gamma) \rightarrow \{ 1, 2, \bot \}$
with finitely many discontinuity points.\footnote{%
The color $\bot$ is redundant for the analysis of the APPs carried out in the
present section.
We introduce it because it makes things simpler in
\Sect{}~\ref{section:analysis-heart} when we employ the APP framework in the
analysis of our online algorithm.
}
For simplicity, in the remainder of this section, we assume that
the APP starts at time
\LongVersion 
$t_{0} = 0$;
this assumption can be lifted by translating any time
$t \in [0, \gamma]$
to
$t + t_{0} \in [t_{0}, t_{0} + \gamma]$.
\LongVersionEnd 
\ShortVersion 
$t_{0} = 0$.
\ShortVersionEnd 

Given some $0 \leq t \leq t' \leq \gamma$, we define the \emph{$1$-volume} and
\emph{$2$-volume} of the interval $[t, t')$ as
\begin{MathMaybe}
V_{1}(t, t') = \int_{t}^{t'} \Indicator(c(x) = 1) \, d x
\end{MathMaybe}
and
\begin{MathMaybe}
V_{2}(t, t') = \int_{t}^{t'} \Indicator(c(x) = 2) \, d x \, ,
\end{MathMaybe}
respectively.
The APP is realized by independent and identically $\ExpDist(\lambda)$
distributed random variables
$Z_{1}, Z_{2}, \dots$
These determine the $[0, \gamma]$-valued random variables
$T_{1}, T_{2}, \dots$,
referred to as \emph{alternation times}, defined inductively by fixing
$T_{0} = 0$
and setting
\[
T_{j}
=
\left\{
\begin{array}{ll}
\max \left\{ t \leq \gamma : V_{1}(T_{j - 1}, t) \leq Z_{j} \right\}, &
\text{$j$ is odd} \\
\max \left\{ t \leq \gamma : V_{2}(T_{j - 1}, t) \leq Z_{j} \right\}, &
\text{$j$ is even}
\end{array}
\right.
\]
for $j = 1, 2, \dots$
Put differently, the alternation times divide the process into
\emph{iterations} so that iteration $j$ lasts from time $T_{j - 1}$ to time
$T_{j}$.
In odd (resp., even) iterations, the process \emph{digests} the $1$s (resp.,
$2$s), ignoring the $\bot$s and the $2$s (resp., $1$s).
If the iteration did not end by time
$T_{j - 1} < t < \gamma$
and
$c(t) = 1$
(resp.,
$c(t) = 2$),
then it ends at time $t + d t$ with probability
$\pi = \lambda d t$;
the iteration ends at time $\lambda$ if it did not end beforehand (an
illustration is provided in \Fig{}~\ref{figure:app}).

\def\FigureApp{
\LongVersion
\begin{figure}
\LongVersionEnd
\ShortVersion
\begin{figure}[h]
\ShortVersionEnd
\begin{center}
\includegraphics[width=\textwidth]{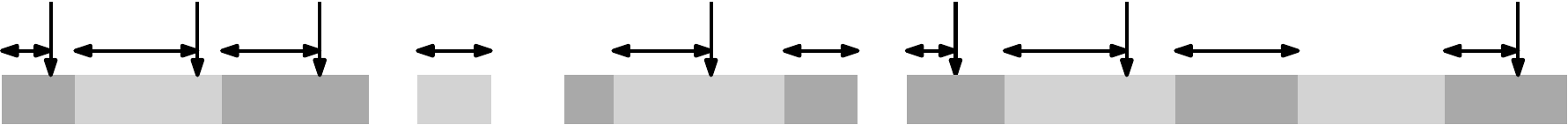}
\end{center}
\caption{ \label{figure:app}
A realization of an alternating Poisson process with time progressing from
left to right.
The dark gray, light gray, and white intervals represent the colors $1$, $2$,
and $\bot$, respectively.
The vertical arrows represent the meaningful alternation times and the
horizontal two-sided arrows depict the time intervals that contribute to the
digestion of the corresponding iterations.
}
\end{figure}
}
\LongVersion 
\FigureApp{}
\LongVersionEnd 

The definition of the alternation times implies, in particular, that if
$T_{j - 1} = \gamma$, then $T_{j} = \gamma$;
we say that the $j$th alternation time is \emph{meaningful} if
$0 < T_{j} < \gamma$.
Observe that if $T_{j}$ is meaningful and $j \geq 1$ is odd (resp., even),
then $c(T_{j})$ must be $1$ (resp., $2$).
Let
\begin{MathMaybe}
N
=
\max \{ j \in \Integers_{\geq 0} \mid T_{j} < \gamma \}
\end{MathMaybe}
be the random variable counting the number of meaningful alternation times.
\LongVersion 
\par
\LongVersionEnd 
Define the $[0, \gamma]$-valued random variables $G_{1}, G_{2}, \dots$ by setting
\[
G_{j}
=
\left\{
\begin{array}{ll}
V_{1}(T_{j - 1}, T_{j}), & \text{$j$ is odd} \\
V_{2}(T_{j - 1}, T_{j}), & \text{$j$ is even}
\end{array}
\right.
\]
and let
$G = \sum_{j = 1}^{\infty} G_{j}$.
We refer to $G_{j}$ as the \emph{digestion} of the $j$th iteration and to $G$
as the \emph{total digestion}.

\begin{lemma} \label{lemma:APP-low-bound-last-digestion}
For every $0 \leq t < \gamma$, we have
$\Expect[G_{j} \mid T_{j - 1} = t]
=
\frac{1}{\lambda} \left( 1 - e^{-\lambda \cdot V_{i}(t, \gamma)} \right)$,
where
$i = 1$
if $j$ is odd;
and
$i = 2$
if $j$ is even.\footnote{%
Recall that for every
$j > 1$,
an odd (resp., even) $j$ implies that
$c(t) = c(T_{j - 1})= 2$
(resp.,
$c(t) = c(T_{j - 1}) = 1$).
}
\end{lemma}
\LongVersion 
\begin{proof}
Assume without loss of generality that $j$ is odd and $i = 1$ (the case that
$j$ is even and $i = 2$ is proved following the same line of arguments).
The design of the APP implies that conditioned on
$T_{j - 1} = t$,
the random variable $G_{j}$ satisfies
$G_{j} \sim \min \{ \ExpDist(\lambda), V_{1}(t, \gamma) \}$,
that is, it is distributed identically to an exponential random variable with
rate $\lambda$, truncated at $V_{1}(t, \gamma)$.
Fixing
$\vartheta = V_{1}(t, \gamma)$,
the assertion follows by observing that
\begin{align*}
\Expect \left[ \min \left\{ \ExpDist(\lambda), \vartheta \right\}
\right]
\, = \, &
\int_{0}^{\vartheta} \lambda e^{-\lambda x} x \, d x \, + \, \vartheta e^{-\lambda
\vartheta} \\
= \, &
\left.
-e^{-\lambda x} x - \frac{1}{\lambda} e^{-\lambda x}
\right|_{0}^{\vartheta}
\, + \, \vartheta e^{-\lambda \vartheta} \\
= \, &
- \vartheta e^{-\lambda \vartheta}
- \frac{1}{\lambda} e^{-\lambda \vartheta}
+ \frac{1}{\lambda}
+ \vartheta e^{-\lambda \vartheta} \\
= \, &
\frac{1}{\lambda} \left( 1 - e^{-\lambda \vartheta} \right) \, ,
\end{align*}
where the second transition is derived using integration by parts with
$u(x) = x$
and
$v(x) = -e^{-\lambda x}$.
\end{proof}
\LongVersionEnd 

\begin{lemma} \label{lemma:APP-digestion-equals-number-meaningful}
$\Expect[G]
=
\Expect[N] / \lambda$.
\end{lemma}
\LongVersion 
\begin{proof}
Let $I_{j}$, $j = 1, 2, \dots$, be an indicator random variable for the event
$T_{j} < \gamma$
and notice that
\[
\Expect[N]
\, = \,
\sum_{j = 1}^{\infty} \Prob \left( N \geq j \right)
\, = \,
\sum_{j = 1}^{\infty} \Expect \left[ I_{j} \right] \, .
\]
Recalling that
\[
\Expect[G]
\, = \,
\sum_{j = 1}^{\infty} \Expect \left[ G_{j} \right] \, ,
\]
it suffices to prove that
$\Expect[I_{j}] / \lambda = \Expect[G_{j}]$
for $j = 1, 2, \dots$
To that end, we show that
\[
\Expect \left[ \Expect \left[ I_{j} \mid T_{j - 1} \right] \right] / \lambda
\, = \,
\Expect \left[ \Expect \left[ G_{j} \mid T_{j - 1} \right] \right] \,
\]
which establishes the assertion by the law of total expectation.

The random variable
$\Expect[I_{j} \mid T_{j - 1}]$
maps the event
$T_{j - 1} = t$
to
\[
\Expect[I_{j} | T_{j - 1} = t]
\, = \,
\Prob(\ExpDist(\lambda) < V_{i}(t, \gamma))
\, = \,
1 - e^{-\lambda \cdot V_{i}(t, \gamma)} \, ,
\]
where
$i = 1$
if $j$ is odd; and
$i = 2$
if $j$ is even.
The proof is completed by \Lem{}~\ref{lemma:APP-low-bound-last-digestion} as
the random variable
$\Expect[G_{j} \mid T_{j - 1}]$
maps the event
$T_{j - 1} = t$
to
$\Expect[G_{j} \mid T_{j - 1} = t]$.
\end{proof}
\LongVersionEnd 

\ShortVersion \sloppy \ShortVersionEnd
\begin{lemma} \label{lemma:APP-bound-number-meaningful}
The random variable $N$ is stochastically dominated by
$1 + 2 Z$,
where
\LongVersion 
$Z \sim \PoisDist(\lambda \cdot \min \{ V_{1}(0, \gamma), V_{2}(0, \gamma) \})$
\LongVersionEnd 
\ShortVersion 
$Z$
\ShortVersionEnd 
is a Poisson random variable with parameter
$\lambda \cdot \min \{ V_{1}(0, \gamma), V_{2}(0, \gamma) \}$.
Moreover, if $K$ denotes the number of discontinuity points of the coloring
function $c$ in $[0, \gamma)$, then
$N \leq K + 1$
(with probability $1$).
\end{lemma}
\ShortVersion \par\fussy \ShortVersionEnd
\LongVersion 
\begin{proof}
Fix
$V_{1} = V_{1}(0, \gamma)$
and
$V_{2} = V_{2}(0, \gamma)$
and define the random variables
\[
N_{1}
\, = \,
\left| \left\{ j \in \Integers_{\geq 0} \mid T_{2 j + 1} < \gamma \right\}
\right|
\quad \text{and} \quad
N_{2}
\, = \,
\left| \left\{ j \in \Integers_{\geq 1} \mid T_{2 j} < \gamma \right\} \right|
\, .
\]
The definition of the APP ensures the following four properties:
\begin{DenseEnumerate}

\item[(P1)]
$N = N_{1} + N_{2}$;

\item[(P2)]
$N_{2} \leq N_{1} \leq N_{2} + 1$;

\item[(P3)]
$N_{i}$, $i \in \{ 1, 2 \}$, is stochastically dominated by $\PoisDist(\lambda
\cdot V_{i})$; and

\item[(P4)]
$N_{i}$, $i \in \{ 1, 2 \}$, is bounded from above by the number of (set-wise)
maximal intervals
$I \subseteq [0, \gamma)$
satisfying
$c(t) = i$
for all $t \in I$.

\end{DenseEnumerate}
The second part of the assertion follows directly from properties (P1) and
(P4).
For the first part, we employ (P1) and (P2) to conclude that
$N \leq 1 + 2 N_{i}$
for $i \in \{ 1, 2 \}$.
Then, by (P3), it follows that $N$ is stochastically dominated by
$1 + 2 \cdot \PoisDist(\lambda \cdot V_{i})$
for $i \in \{ 1, 2 \}$,
thus it is stochastically dominated by
$1 + 2 \cdot \PoisDist(\lambda \cdot \min\{ V_{1}, V_{2} \})$.
\end{proof}
\LongVersionEnd 

\LongVersion 
It will be convenient to also consider a generalization of the APP, referred to
as a \emph{rate-varying APP}, in which the fixed rate parameter $\lambda$ is
replaced by a \emph{rate function}
$\lambda' : [0, \gamma) \rightarrow \Reals_{> 0}$
that may vary in time.
This affects the aforementioned iteration termination probability $\pi$ so
that an odd (resp., even) iteration $j$ that did not end by time
$T_{j - 1} < t < \gamma$,
$c(t) = 1$
(resp.,
$c(t) = 2$),
will now end at time $t + d t$ with probability
$\pi = \pi(t) = \lambda'(t) d t$.
Given some (fixed)
$\lambda \in \Reals_{> 0}$,
it is straightforward to verify that if the rate function $\lambda'(t)$ is
bounded from above by $\lambda$, i.e.,
$\lambda'(t) \leq \lambda$
for all
$0 \leq t < \lambda$,
then \Lem{} \ref{lemma:APP-low-bound-last-digestion} and
\ref{lemma:APP-bound-number-meaningful} hold also for rate-varying APPs,
only that in the former, we should replace the equality in
$\Expect[G_{j} \mid T_{j - 1} = t]
=
\frac{1}{\lambda} \left( 1 - e^{-\lambda \cdot V_{i}(t, \gamma)} \right)$
with a $\geq$ inequality.
\LongVersionEnd 

\begin{IntuitionSpotlight}
APPs are utilized in the analysis conducted in
\Sect{}~\ref{section:analysis-heart} as they capture the behavior of the
stilt-walker algorithm in what can be informally described as ``toggling
situations''.
Such situations turn out to appear in multiple parts of the analysis (see
\Lem{} \ref{lemma:phase-app}, \ref{lemma:bound-time-cost-0-subphase}, and
\ref{lemma:partition-time-line-into-phases}).
\end{IntuitionSpotlight}

\subsection{Analysis under the end-of-input signal assumption}
\label{section:analysis-heart}
Let $\mathcal{T}$ be an $n$-point $\alpha$-HSBT of aspect ratio $\Delta$
and let $T$ and
$w : T \rightarrow \Reals_{\geq 0}$
be the full binary tree and weight function that realize $\mathcal{T}$.
Assume without loss of generality that the minimum positive distance in
$\mathcal{T}$ is scaled to $1$ so that $\Delta$ is the diameter of
$\mathcal{T}$.

Our goal in this section is to establish \Thm{}~\ref{theorem:main} under the
end-of-input signal assumption.\footnote{%
For the convenience of the reader, \Fig{}~\ref{figure:claim-chart} provides a
schematic overview of the analysis presented in this section.
}
More formally, assume that the online algorithm is signaled at time
$\EndTime = \max \{ \aTime(\rho) \mid \rho \in R \}$
\NotationLabel{analysis:end-time}
(the arrival time of the last request in $R$);
upon receiving this signal, the algorithm clears the remaining active
requests by immediately matching across $v$ for every effective vertex
$v \in \Effective(\EndTime)$
(this is guaranteed as the number of active requests at time $\EndTime$ must
be even).
Let
$\sEndCost$
\NotationLabel{analysis:end-cost}
be the space cost of these matching operations and observe that
$\sEndCost \leq (n / 2) \cdot \Delta$.
\LongVersion 
(Although it does not affect our analysis, it is interesting to point out that
$\sEndCost$ is, in fact, the cost of an optimal matching of the remaining
requests.)
\LongVersionEnd 
We start the analysis with the following ``warmup''
\LongVersion 
observation regarding the operation of the stilt-walker algorithm.
\LongVersionEnd 
\ShortVersion 
observation.
\ShortVersionEnd 

\begin{observation*}
Consider an internal vertex $v \in T - \Leaves$ with children $u_{1}, u_{2}$.
The design of $\A$ ensures that:
\begin{DenseEnumerate}

\item \label{item:independence}
the random variable
$\Indicator(v \in \Odd(t))$
is independent of the coin tosses of all vertices $u \in T(v)$ (including
$v$);

\item \label{item:match-on-top}
$\A$ can match on top of $v$ only when $v$ is odd; and

\item \label{item:match-across}
if $\A$ matched across or on top of $v$ at time $t$, then
\LongVersion 
$v$, $u_{1}$, and $u_{2}$ are not odd immediately following time
$t$, i.e.,
\LongVersionEnd 
$v, u_{1}, u_{2} \notin \Odd(t + d t)$
for infinitesimally small
$d t > 0$.

\end{DenseEnumerate}
\end{observation*}
\begin{proof}
To establish property \ref{item:independence}, notice that the coin tosses of
vertex $u$ determine the decisions of $\A$ to match across $u$.
Matching across $u$ decreases $|\Active_{v}(t)|$ by $2$, hence it does not
affect its parity.
\LongVersion \par \LongVersionEnd
Property \ref{item:match-on-top} is proved by recalling that matching on top
of $v$ at time $t$ is realized by matching a request located in some leaf
$x \in \Leaves(v)$
to a request located in some leaf
$x' \in \Leaves - \Leaves(v)$.
Since $v \neq \LCA(x, x')$, it must belong to the stilt in $\Stilts(t)$
whose foot is $x$ which establishes the assertion by the definition of
$\Stilts(t)$.
\LongVersion \par \LongVersionEnd
Finally, observe that property \ref{item:match-across} holds trivially if $\A$
matched across $v$ at time $t$ because this means that
$u_{1}, u_{2} \in \Odd(t)$
and thus,
$v, u_{1}, u_{2} \notin \Odd(t + d t)$.
Otherwise, if $\A$ matched on top of $v$ at time $t$, then
$v \in \Odd(t)$
which means that
$u_{i} \in \Odd(t)$
and
$u_{3 - i} \notin \Odd(t)$
for some
$i \in \{ 1, 2 \}$.
This also means that $\A$ matched on top of $u_{i}$ at time $t$, therefore
$v, u_{i}, u_{3 - i} \notin \Odd(t + d t)$.
\end{proof}

\begin{IntuitionSpotlight}
A key ingredient in the analysis of $\A$'s competitive ratio is an alternative
method for measuring its time and space cost on a \emph{per-vertex basis}.
This is facilitated by the definitions of time and space potentials for each
internal vertex $v$.
\end{IntuitionSpotlight}

\paragraph{Time and space potentials.}
Consider some internal vertex $v \in T - \Leaves$ with children $u_{1},
u_{2}$ and some
$0 \leq t_{0} < t_{1} \leq \EndTime$.
The \emph{time potentials} of $v$, denoted $\tau_{v}$
and $\adv{\tau}_{v}$,
capture the contributions of $v$ to $\tCost_{\A}(R, \mathcal{T})$ and
$\tCost_{\adv{\A}}(R, \mathcal{T})$, respectively, in a certain time interval.
They are defined by setting
\[
\tau_{v}([t_{0}, t_{1}))
\, = \,
\int_{t_{0}}^{t_{1}} \Indicator(v \in \Effective(t)) \, d t
\quad \text{and} \quad
\adv{\tau}_{v}([t_{0}, t_{1}))
\, = \,
\int_{t_{0}}^{t_{1}} \Indicator(u_{1} \in \adv{\Odd}(t)) + \Indicator(u_{2}
\in \adv{\Odd}(t)) \, d t \, ;
\]
\NotationLabel{analysis:tau}
\NotationLabel{analysis:adv-tau}
in other words, a $d t$ amount is deposited into $\tau_{v}$ whenever $v \in
\Effective(t)$ and into $\adv{\tau}_{v}$ whenever $u_{i} \in \adv{\Odd}(t)$
for $i \in \{ 1, 2 \}$.

\ShortVersion \sloppy \ShortVersionEnd
The \emph{space potentials} of $v$, denoted $\sigma_{v}$
\NotationLabel{analysis:sigma}
and $\adv{\sigma}_{v}$,
\NotationLabel{analysis:adv-sigma}
capture the contributions of $v$ to $\sCost_{\A}(R, \mathcal{T})$ and
$\sCost_{\adv{\A}}(R, \mathcal{T})$, respectively, in a certain time interval.
An amount of $w(v)$ is deposited into $\sigma_{v}$ whenever $\A$ matches
across $v$;
an amount of $w(v)$ is deposited into $\adv{\sigma}_{v}$ whenever
$\adv{\A}$ matches across or on top of $v$.
In other words, given two requests $\rho, \rho'$ with
$x = \Location(\rho)$
and
$x' = \Location(\rho')$,
if $\A$ matches requests $\rho$ and $\rho'$,
then we deposit an amount of $w(u)$ into $\sigma_{u}$ for
$u = \LCA(x, x')$;
if $\adv{\A}$ matches requests $\rho$ and $\rho'$,
then we deposit an amount of $w(u)$ into $\adv{\sigma}_{u}$ for every internal
vertex $u$ along the unique path connecting $x$ and $x'$ in $T$.
Let $\sigma_{v}([t_{0}, t_{1}))$ and $\adv{\sigma}_{v}([t_{0}, t_{1}))$ be the
total amount deposited into $\sigma_{v}$ and $\adv{\sigma}_{v}$, respectively,
during the time interval $[t_{0}, t_{1})$.
\ShortVersion \par\fussy \ShortVersionEnd

For clarity of the exposition, we often write $\tau_{v}(t_{0}, t_{1})$,
$\adv{\tau}_{v}(t_{0}, t_{1})$, $\sigma_{v}(t_{0}, t_{1})$, and
$\adv{\sigma}_{v}(t_{0}, t_{1})$ instead of the aforementioned notations.
We also extend the definition of these four notations from intervals to
collections of disjoint intervals in the natural manner.
\Thm{}~\ref{theorem:main} is established by proving the following three
lemmas.

\begin{IntuitionSpotlight}
\Lem{}~\ref{lemma:expressing-costs-by-potentials} allows us to
express the time and space costs by means of the per-vertex potentials.
\Lem{}~\ref{lemma:bounding-time-potential} then means that we can
bound the time potential of $v$ under $\A$ by the time and space potentials of
$v$ under $\adv{\A}$, charging the extra $w(v)$ on the additive term of the
competitive ratio, whereas \Lem{}~\ref{lemma:bounding-space-potential} means
that we can bound the space potential of $v$ under $\A$ by its time
potential.
\end{IntuitionSpotlight}

\begin{lemma} \label{lemma:expressing-costs-by-potentials}
There exists some
$\zeta = \zeta(R)$
such that the time potentials satisfy
\[
\tCost_{\A}(R, \mathcal{T})
\leq
\zeta + \sum_{v \in T - \Leaves} O (\tau_{v}(0, \EndTime))
\quad \text{and} \quad
\tCost_{\adv{\A}}(R, \mathcal{T})
\geq
\zeta / h + \sum_{v \in T - \Leaves} \Omega (\adv{\tau}_{v}(0, \EndTime) / h)
\]
(recall that $h$ denotes the height of $T$).
The space potentials satisfy
\[
\sCost_{\A}(R, \mathcal{T})
\leq
\sEndCost + \sum_{v \in T - \Leaves} O (\sigma_{v}(0, \EndTime))
\quad \text{and} \quad
\sCost_{\adv{\A}}(R, \mathcal{T})
\geq
\sum_{v \in T - \Leaves} \Omega ((\alpha - 1) \cdot \adv{\sigma}_{v}(0, \EndTime))
\]
(recall that the parameter $\alpha$ is set in \Thm{}~\ref{theorem:main}).
\end{lemma}

\begin{lemma} \label{lemma:bounding-time-potential}
For every $v \in T - \Leaves$, it holds that
$\Expect[\tau_{v}(0, \EndTime)]
\leq
O (\adv{\tau}_{v}(0, \EndTime) + \adv{\sigma}_{v}(0, \EndTime) + w(v))$.
\end{lemma}

\begin{lemma} \label{lemma:bounding-space-potential}
For every $v \in T - \Leaves$, it holds that
$\Expect[\sigma_{v}(0, \EndTime)]
\leq
\Expect[\tau_{v}(0, \EndTime)]$.
\end{lemma}

\begin{proof}[Proof of \Lem{}~\ref{lemma:expressing-costs-by-potentials}]
We first note that
\begin{MathMaybe}
\tCost_{\A}(R, \mathcal{T})
=
\sum_{v \in T} \int_{0}^{\EndTime} \Indicator(v \in \Heads(t)) \, d t \, .
\end{MathMaybe}
Indeed, as each leaf contains at most one active request, an active request
$\rho \in \Active(t)$ is accounted for in exactly one term of the sum in the
RHS of the equation, that is, the term corresponding to the head
of the stilt whose foot is $\Location(\rho)$.
Since an internal vertex is effective at time $t$ if and only if its two
children are in $\Heads(t)$, the last equation can be rewritten as
\[
\tCost_{\A}(R, \mathcal{T})
=
\int_{0}^{\EndTime} \Indicator(r \in \Odd(t)) \, d t
+
2 \cdot \sum_{v \in T - \Leaves} \tau_{v}(0, \EndTime) \, .
\]
On the other hand, the inequality
\begin{MathMaybe}
\tCost_{\adv{\A}}(R, \mathcal{T})
\geq
\frac{1}{h} \cdot \sum_{v \in T} \int_{0}^{\EndTime} \Indicator(v \in
\adv{\Odd}(t)) \, d t
\end{MathMaybe}
holds since each active request under $\adv{\A}$ is accounted for in at
most $h$ terms of the sum in the RHS of the inequality, therefore
\[
\tCost_{\adv{\A}}(R, \mathcal{T})
\geq
\frac{1}{h} \left(
\int_{0}^{\EndTime} \Indicator(r \in \adv{\Odd}(t)) \, d t
+
\sum_{v \in T - \Leaves} \adv{\tau}_{v}(0, \EndTime)
\right) \, .
\]
The first part of the assertion is established by observing that
$r \in \Odd(t)$
if and only if
$r \in \adv{\Odd}(t)$,
hence we can fix
\begin{MathMaybe}
\zeta
=
\int_{0}^{\EndTime} \Indicator(r \in \Odd(t)) \, d t
=
\int_{0}^{\EndTime} \Indicator(r \in \adv{\Odd}(t)) \, d t \, .
\end{MathMaybe}

The contribution to $\sCost_{\A}(R, \mathcal{T})$ of matching requests $\rho$
and $\rho'$ by $\A$ is $w(\LCA(x, x'))$;
this is also its contribution to the space potentials $\sigma$, hence
\[
\sCost_{\A}(R, \mathcal{T})
=
\sEndCost
+
\sum_{v \in T - \Leaves} \sigma_{v}(0, \EndTime) \, .
\]
The contribution to
$\sCost_{\adv{\A}}(R, \mathcal{T})$
of matching requests $\rho$ and $\rho'$ by $\adv{\A}$ is $w(\LCA(x, x'))$,
whereas since
$\mathcal{T} = (T, w)$
is an $\alpha$-HSBT (recall that
$1 < \alpha \leq 2$),
its contribution to the space potentials $\adv{\sigma}$ is bounded
from above by
$\sum_{i = 0}^{h} w(\LCA(x, x')) \cdot (1 / \alpha)^{i}
<
w(\LCA(x, x')) \cdot \alpha / (\alpha - 1)$,
hence,
\[
\sCost_{\adv{\A}}(R, \mathcal{T})
\geq
\Omega (\alpha - 1) \cdot \sum_{v \in T - \Leaves}
\adv{\sigma}_{v}(0, \EndTime)
\]
which completes the proof.
\end{proof}

\paragraph{Convenient notation.}
The remainder of this section is dedicated to the proofs of \Lem{}
\ref{lemma:bounding-time-potential} and \ref{lemma:bounding-space-potential}.
To this end, we fix some internal vertex $v \in T - \Leaves$ with children
$u_{1}$ and $u_{2}$ which facilitates switching to a shorter and simpler
notation:
Denote
$\tau = \tau_{v}$,
$\adv{\tau} = \adv{\tau}_{v}$,
$\sigma = \sigma_{v}$, and
$\adv{\sigma} = \adv{\sigma}_{v}$.
Given some time $t \in [0, \EndTime)$, we write for short
\[
X_{i}(t) = \Indicator(u_{i} \in \Odd(t))
\qquad
\adv{X}_{i}(t) = \Indicator(u_{i} \in \adv{\Odd}(t))
\]
\NotationLabel{analysis:variable-X-i}
\NotationLabel{analysis:variable-adv-X-i}
for $i \in \{ 1, 2 \}$ and
\[
X(t) = X_{1}(t) \xor X_{2}(t)
\qquad
\adv{X}(t) = \adv{X}_{1}(t) \xor \adv{X}_{2}(t) \, .
\]
\NotationLabel{analysis:variable-X}
\NotationLabel{analysis:variable-adv-X}
\LongVersion 
Notice that
\LongVersionEnd 
\ShortVersion 
As
\ShortVersionEnd 
$\Indicator(v \in \Effective(t))
=
X_{1}(t) \cdot X_{2}(t)$
and
$\Indicator(u_{1} \in \adv{\Odd}(t)) + \Indicator(u_{2} \in \adv{\Odd}(t))
=
\adv{X}(t) + 2 \cdot \adv{X}_{1}(t) \cdot \adv{X}_{2}(t)$,
\LongVersion 
thus
\LongVersionEnd 
\ShortVersion 
we get
\ShortVersionEnd 
\[
\tau(t_{0}, t_{1})
=
\int_{t_{0}}^{t_{1}} X_{1}(t) \cdot X_{2}(t) \, d t 
\quad \text{and} \quad
\adv{\tau}(t_{0}, t_{1})
=
\int_{t_{0}}^{t_{1}} \adv{X}(t) + 2 \cdot \adv{X}_{1}(t) \cdot \adv{X}_{2}(t) \,
d t \, .
\]

It will be convenient to also define
\begin{MathMaybe}
Y_{i}(t)
=
|\{ \rho \in R \mid \Location(\rho) \in \Leaves(u_{i}) \land
\aTime(\rho) \leq t \}| \pmod{2}
\end{MathMaybe}
\NotationLabel{analysis:variable-Y-i}
for $i \in \{ 1, 2 \}$ and
\begin{MathMaybe}
Y(t) = Y_{1}(t) \xor Y_{2}(t) \, ,
\end{MathMaybe}
\NotationLabel{analysis:variable-Y}
observing that the parity of the number of times $\A$
matched on top of $u_{i}$ (resp., $v$) up to time $t$ equals
$X_{i}(t) \xor Y_{i}(t)$
(resp.,
$X(t) \xor Y(t)$).

\paragraph{Phases and subphases.}
We partition the time line
$[0, \EndTime)$
into \emph{phases} (defined with respect to $v$),
where each phase is a time interval that starts when the previous phase
ends (or at time $0$ if this is the first phase) and ends when $\A$ matches on
top of $v$ (or at time $\EndTime$ if this is the last phase).
A crucial observation is that this partition is fully determined by the coin
tosses of $\Ancestors(v)$ (namely, the ancestors of $v$) independently of the
coin tosses of $v$.

We further partition every phase
$\phi = [t_{0}, t_{1})$
of $v$ into \emph{subphases}, where each subphase is
a time interval that starts when the previous subphase ends (or at time
$t_{0}$ if this is the first subphase of $\phi$) and ends when
$\adv{\A}$ matches across \emph{or} on top of $v$ (or at time $t_{1}$ if this
is the last subphase of $\phi$).
Notice that matching operations across $v$ performed by $\A$ (fully
determined by the coin tosses of $v$) can occur at the midst of a subphase.

\begin{lemma} \label{lemma:phase-app}
For every phase
$\phi = [t_{0}, t_{1})$
of $v$, it holds that
$\Expect_{v}[\sigma(\phi)]
=
\Expect_{v}[\tau(\phi)]$.
\end{lemma}
\LongVersion 
\begin{proof}
We investigate the dynamics of
$(X_{1}(t), X_{2}(t))_{t \in \phi}$
and
$(Y_{1}(t), Y_{2}(t))_{t \in \phi}$
that take values in
$\{ 0, 1 \}^{2}$
(an illustration is provided in \Fig{}~\ref{figure:phase-app}).
Observe that a new request arriving in $\Leaves(u_{i})$, $i \in \{ 1, 2 \}$,
flips $X_{i}$ and $Y_{i}$ without affecting $X_{3 - i}$ and $Y_{3 - i}$.
While $(Y_{1}, Y_{2})$ is affected only by new request arrivals, the dynamic
of $(X_{1}, X_{2})$ is tied to the actions of $\A$ too.
Specifically, $\A$ can match across $v$ (recall that $\A$ does not match on
top of $v$ in the midst of phase $\phi$) only when
$(X_{1}, X_{2}) = (1, 1)$
and if
$(X_{1}, X_{2}) = (1, 1)$
throughout the infinitesimally small time interval $[t - d t, t)$,
then $\A$ matches across $v$ at time $t$ with probability
$d t / w(v)$
(depending solely on the coin tosses of $v$),
in which case
$(X_{1}, X_{2})$
flips to
$(X_{1}(t), X_{2}(t)) = (0, 0)$.
Moreover, we know that
$(X_{1}(t_{0}), X_{2}(t_{0})) = (0, 0)$.

Let
$(y_{1}, y_{2}) = (Y_{1}(t_{0}), Y_{2}(t_{0}))$.
We color the times in $\phi$ using the coloring function
$c : \phi \rightarrow \{ 1, 2, \bot \}$
by setting
\[
c(t)
=
\left\{
\begin{array}{ll}
1, & (Y_{1}(t), Y_{2}(t)) = (\neg y_{1}, \neg y_{2}) \\
2, & (Y_{1}(t), Y_{2}(t)) = (y_{1}, y_{2}) \\
\bot, & \text{o.w.}
\end{array}
\right.
\]
The key observation now is that the times at which $\A$ matches across $v$
can be viewed as the meaningful alternation times of an APP $\Pi_{\phi}$
defined over the time interval $\phi$ with coloring function $c(\cdot)$ and
rate $1 / w(v)$.
(Notice that the role of $(y_{1}, y_{2})$ in the validity of this observation
is simply to adjust the dynamic of $(X_{1}, X_{2})$, starting with
$(X_{1}(t_{0}), X_{2}(t_{0})) = (0, 0)$,
to the APP framework in which the first digested color is defined to be $1$.)
Taking $N$ to be the random variable counting the number of meaningful
alternation times in $\Pi_{\phi}$ and $G$ to be its total digestion, we
conclude that
$\sigma(\phi)
=
w(v) \cdot N$
and
$\tau(\phi)
=
G$.
The assertion follows by
\Lem{}~\ref{lemma:APP-digestion-equals-number-meaningful}.
\end{proof}
\LongVersionEnd 

\def\FigurePhaseApp{
\LongVersion
\begin{figure}
\LongVersionEnd
\ShortVersion
\begin{figure}[h]
\ShortVersionEnd
\begin{center}
\includegraphics[width=\textwidth]{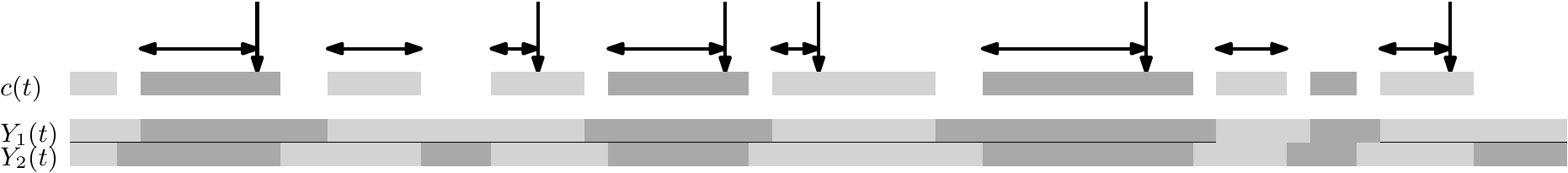}
\end{center}
\caption{ \label{figure:phase-app}
Phase $\phi$ with time progressing from left to right, assuming that $(y_{1},
y_{2}) = (0, 0)$.
Bottom rows:
the dark gray and light gray intervals represent the times $t$ at which
$Y_{i}(t) = 1$
and
$Y_{i}(t) = 0$,
respectively.
Top row:
the dark gray, light gray, and white intervals represent the times $t$ at
which
$c(t) = 1$,
$c(t) = 2$,
and
$c(t) = \bot$,
respectively.
The vertical arrows represent the times at which $\A$ matches across $v$ and
the horizontal two-sided arrows depict the time intervals that contribute to
$\tau(\phi)$, i.e., when
$(X_{1}, X_{2}) = (1, 1)$.
Notice that towards $\phi$'s end, we must have
$X = X_{1} \xor X_{2} = 1$
unless $\phi$ is the last phase.
}
\end{figure}
}
\LongVersion 
\FigurePhaseApp{}
\LongVersionEnd 

Fixing the coin tosses in $\Ancestors(v)$ and thus, fixing the partition of
$[0, \EndTime)$ into phases, we can apply \Lem{}~\ref{lemma:phase-app}
to the each individual phase, thus establishing
\Lem{}~\ref{lemma:bounding-space-potential} by the linearity of expectation.
The remainder of this section is dedicated to proving
\Lem{}~\ref{lemma:bounding-time-potential}.
The first step towards achieving this goal is to bound the time potential of
$\A$ per subphase based on the following subphase classification.

\paragraph{$0$- and $1$-subphases.}
Fix some subphase $\varphi$ of $v$.
Notice that matching across $v$ (by $\A$) does not affect $\adv{X}_{i}$, $i
\in \{ 1, 2 \}$, nor does it change
$X_{1} \xor X_{2}$.
Thus, there exists some
$b = b(\varphi) \in \{ 0, 1 \}$
such that
$X_{1}(t) \xor X_{2}(t) \xor \adv{X}_{1}(t) \xor \adv{X}_{2}(t) = b$
for all
$t \in \varphi$;
in what follows, we distinguish between two types of subphases:
\emph{$0$-subphases}, for which $b = 0$, and \emph{$1$-subphases}, for which
$b = 1$.

\begin{observation} \label{observation:bound-time-cost-1-subphase}
If $\varphi$ is a $1$-subphase, then
$\tau(\varphi)
\leq
\adv{\tau}(\varphi)$.
\end{observation}
\begin{proof}
Recall that
$\tau(\varphi)
=
\int_{\varphi} X_{1}(t) \cdot X_{2}(t) d t$
and
$\adv{\tau}(\varphi)
\geq
\int_{\varphi} \adv{X}_{1}(t) \xor \adv{X}_{2}(t) d t$.
The assertion follows by the definition of a $1$-subphase ensuring that for
every
$t \in \varphi$,
if
$(X_{1}(t), X_{2}(t)) = (1, 1)$,
then
$(\adv{X}_{1}(t), \adv{X}_{2}(t)) \in \{ (0, 1), (1, 0) \}$.
\end{proof}

\begin{lemma} \label{lemma:bound-time-cost-0-subphase}
If $\varphi$ is a $0$-subphase, then
$\Expect_{v}[\tau(\varphi)]
\leq
\adv{\tau}(\varphi) + w(v)$.
\end{lemma}
\ShortVersion \sloppy \ShortVersionEnd
\begin{proof}
We investigate the dynamics of
$(X_{1}(t), X_{2}(t))_{t \in \varphi}$
and
$(\adv{X}_{1}(t), \adv{X}_{2}(t))_{t \in \varphi}$
that take values in
$\{ 0, 1 \}^{2}$
(an illustration is provided in \Fig{}~\ref{figure:subphase-app}).
By the definition of a $0$-subphase, at any time
$t \in \varphi$,
either
$(X_{1}(t), X_{2}(t)) = (\adv{X}_{1}(t), \adv{X}_{2}(t))$
or
$(X_{1}(t), X_{2}(t)) = (\neg\adv{X}_{1}(t), \neg\adv{X}_{2}(t))$;
we refer to the former (resp., latter) as an \emph{agreement} (resp.,
\emph{disagreement}) state of $\A$ and $\adv{\A}$.

Observe that a new request arriving in $\Leaves(u_{i})$, $i \in \{ 1, 2 \}$,
flips $X_{i}$ and $\adv{X}_{i}$ without affecting $X_{3 - i}$ and $\adv{X}_{3
- i}$.
While $(\adv{X}_{1}, \adv{X}_{2})$ is affected only by new request arrivals
(recall that $\adv{\A}$ does not match across or on top of $v$ in the midst of
subphase $\varphi$), the dynamic of $(X_{1}, X_{2})$ is tied to the actions of
$\A$ too.
Specifically, $\A$ can match across $v$ (recall that $\A$ does not match on
top of $v$ in the midst of subphase $\varphi$) only when
$(X_{1}, X_{2}) = (1, 1)$
and if 
$(X_{1}, X_{2}) = (1, 1)$
throughout the infinitesimally small time interval
$[t - d t, t)$,
then $\A$ matches across $v$ at time $t$ with probability
$d t / w(v)$
(depending solely on the coin tosses of $v$), in which case $(X_{1}, X_{2})$
flips to
$(X_{1}(t), X_{2}(t)) = (0, 0)$,
thus toggling the agreement/disagreement state.

Define the functions
$c_{\text{agree}} : \varphi \rightarrow \{ 1, 2, \bot \}$
and
$c_{\text{disagree}} : \varphi \rightarrow \{ 1, 2, \bot \}$
as follows:
\[
c_{\text{agree}}(t)
=
\left\{
\begin{array}{ll}
1, & (\adv{X}_{1}(t), \adv{X}_{2}(t)) = (1, 1) \\
2, & (\adv{X}_{1}(t), \adv{X}_{2}(t)) = (0, 0) \\
\bot, & \text{o.w.}
\end{array}
\right.
\qquad
c_{\text{disagree}}(t)
=
\left\{
\begin{array}{ll}
1, & (\adv{X}_{1}(t), \adv{X}_{2}(t)) = (0, 0) \\
2, & (\adv{X}_{1}(t), \adv{X}_{2}(t)) = (1, 1) \\
\bot, & \text{o.w.}
\end{array}
\right. \, .
\]
We color the times in $\varphi$ using the coloring function
$c : \varphi \rightarrow \{ 1, 2, \bot \}$
by setting
$c = c_{\text{agree}}$
if the subphase starts in an agreement
\LongVersion 
state; and
\LongVersionEnd 
\ShortVersion 
state; and
\ShortVersionEnd 
$c = c_{\text{disagree}}$
\LongVersion 
if the subphase starts in a disagreement state.
\LongVersionEnd 
\ShortVersion 
otherwise.
\ShortVersionEnd 
The key observation now is that the times at which $\A$ matches across $v$
can be viewed as the meaningful alternation times of an APP $\Pi_{\varphi}$
defined over the time interval $\varphi$ with coloring function $c(\cdot)$ and
rate $1 / w(v)$.
(Notice that the role of the $c_{\text{agree}}$ vs.\ $c_{\text{disagree}}$
distinction in the validity of this observation is simply to adjust the
dynamic of $(X_{1}, X_{2})$, starting in an agreement/disagreement state, to
the APP framework in which the first digested color is defined to be $1$.)

Taking $G$ to be the total digestion of $\Pi_{\varphi}$, we notice that
$\tau(\varphi) = G$.
Moreover, the construction of the coloring function $c(\cdot)$ ensures that
$\adv{\tau}(\varphi)
\geq
2 \int_{\varphi} \adv{X}_{1}(t) \cdot \adv{X}_{2}(t) d t
\geq
2 \min \{ V_{1}, V_{2} \}$,
where $V_{1}$ and $V_{2}$ are the total $1$- and $2$-volumes of $\Pi_{\varphi}$,
respectively.
The assertion follows by \Lem{}\
\ref{lemma:APP-digestion-equals-number-meaningful} and
\ref{lemma:APP-bound-number-meaningful}.
\end{proof}
\ShortVersion \par\fussy \ShortVersionEnd

\def\FigureSubphaseApp{
\LongVersion
\begin{figure}
\LongVersionEnd
\ShortVersion
\begin{figure}[h]
\ShortVersionEnd
\begin{center}
\includegraphics[width=\textwidth]{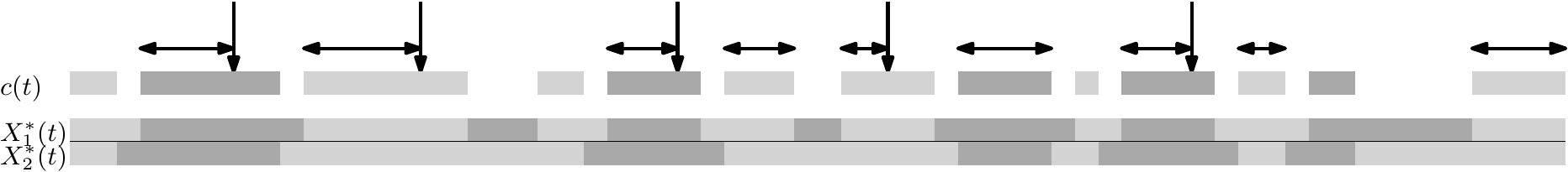}
\end{center}
\caption{ \label{figure:subphase-app}
Subphase $\varphi$ with time progressing from left to right, assuming that the
subphase starts in an agreement state.
Bottom rows:
the dark gray and light gray intervals represent the times $t$ at which
$\adv{X}_{i}(t) = 1$
and
$\adv{X}_{i}(t) = 0$,
respectively.
Top row:
the dark gray, light gray, and white intervals represent the times $t$ at
which
$c(t) = 1$,
$c(t) = 2$,
and
$c(t) = \bot$,
respectively.
The vertical arrows represent the times at which $\A$ matches across $v$ and
the horizontal two-sided arrows depict the time intervals that contribute to
$\tau(\phi)$, i.e., when
$(X_{1}, X_{2}) = (1, 1)$.
Notice that by the definition of $\adv{\tau}$, times $t$ at which
$\adv{X}_{1}(t) \xor \adv{X}_{2}(t) = 1$
(marked as white intervals in the top row)
also contribute to $\adv{\tau}(\phi)$, but this contribution is ignored by our
analysis.
}
\end{figure}
}
\LongVersion 
\FigureSubphaseApp{}
\LongVersionEnd 

\paragraph{$0$- and $1$-phases.}
Phase $\phi$ of $v$ is said to be a \emph{$0$-phase} (resp., a
\emph{$1$-phase}) if it starts with a $0$-subphase (resp., a $1$-subphase).
Let
$P^{0}$
\NotationLabel{analysis:P-0}
(resp.,
$P^{1}$)
\NotationLabel{analysis:P-1}
be the set of $0$-phases (resp., $1$-phases) of
$v$.
Using \Obs{}~\ref{observation:bound-time-cost-1-subphase} and
\Lem{}~\ref{lemma:bound-time-cost-0-subphase}, we establish
\Lem{}~\ref{lemma:bounding-time-potential} (our goal in the remainder of this
section) by proving the following inequalities:
\begin{align}
&
\Expect_{v, \Ancestors(v)} \left[ \tau(P^{0}) \right]
\leq
O \left( \adv{\tau}(0, \EndTime) + \adv{\sigma}(0, \EndTime) + w(v) \right)
\label{equation:target-bound-0-phases} \\
&
\Expect_{v, \Ancestors(v)} \left[ \tau(P^{1}) \right]
\leq
O \left( \adv{\tau}(0, \EndTime) + \adv{\sigma}(0, \EndTime) \right) \, .
\label{equation:target-bound-1-phases}
\end{align}
\Lem{}~\ref{lemma:0-and-1-phases} (a combination of
\Obs{}~\ref{observation:bound-time-cost-1-subphase} and
\Lem{}~\ref{lemma:bound-time-cost-0-subphase} essentially) plays an important
role in the desired proofs.

\begin{lemma} \label{lemma:0-and-1-phases}
If $\phi$ is a $0$-phase, then
\begin{MathMaybe}
\Expect_{v} \left[ \tau(\phi) \right]
\leq
\adv{\tau}(\phi) + 2 \adv{\sigma}(\phi) + w(v) \, ;
\end{MathMaybe}
if $\phi$ is a $1$-phase, then
\begin{MathMaybe}
\Expect_{v} \left[ \tau(\phi) \right]
\leq
\adv{\tau}(\phi) + 2 \adv{\sigma}(\phi) \, .
\end{MathMaybe}
\end{lemma}
\LongVersion 
\begin{proof}
Let
$U^{b}(\phi)$
be the set of $b$-subphases of $\phi$ for $b \in \{ 0, 1 \}$.
If
$U^{0}(\phi) = \emptyset$
and
$U^{1}(\phi) = \{ \varphi \}$,
then we can employ \Obs{}~\ref{observation:bound-time-cost-1-subphase} to
conclude that
$\Expect_{v}[\tau(\phi)]
\leq
\adv{\tau}(\varphi)
=
\adv{\tau}(\phi)$.
If
$U^{0}(\phi) = \{ \varphi \}$
and
$U^{1}(\phi) = \emptyset$,
then we can employ \Lem{}~\ref{lemma:bound-time-cost-0-subphase} to conclude
that
$\Expect_{v}[\tau(\phi)]
\leq
w(v) + \adv{\tau}(\varphi)
=
w(v) + \adv{\tau}(\phi)$.

Since all but the last subphases of $\phi$ end when $\adv{\A}$ matches across
or on top of $v$, it follows by the definition of $\adv{\sigma}$ that
$|U^{0}(\phi) \cup U^{1}(\phi)| = 1 + \adv{\sigma}(\phi) / w(v)$.
Therefore, if
$|U^{0}(\phi) \cup U^{1}(\phi)| > 1$,
then we can employ \Obs{}~\ref{observation:bound-time-cost-1-subphase} and
\Lem{}~\ref{lemma:bound-time-cost-0-subphase} to conclude that
\begin{align*}
\Expect_{v} \left[ \tau(\phi) \right]
\leq &
\sum_{\varphi \in U^{0}(\phi)} \left( w(v) + \adv{\tau}(\varphi) \right)
+
\sum_{\varphi \in U^{1}(\phi)} \adv{\tau}(\varphi) \\
= & \,
\adv{\tau}(\phi) + |U^{0}(\phi)| \cdot w(v) \\
\leq & \,
\adv{\tau}(\phi) + |U^{0}(\phi) \cup U^{1}(\phi)| \cdot w(v) \\
= & \,
\adv{\tau}(\phi)
+
\left( 1 + \adv{\sigma}(\phi) / w(v) \right) \cdot w(v) \\
= & \,
\adv{\tau}(\phi) + w(v) + \adv{\sigma}(\phi)
\, \leq \,
\adv{\tau}(\phi) + 2 \cdot \adv{\sigma}(\phi) \, ,
\end{align*}
where the last transition holds since
$|U^{0}(\phi) \cup U^{1}(\phi)| > 1$
implies that
$\adv{\sigma}(\phi) \geq w(v)$.
The assertion follows.
\end{proof}
\LongVersionEnd 

Fixing the coin tosses in $\Ancestors(v)$ (and thus, fixing the partition of
$[0, \EndTime)$ into phases), we can apply \Lem{}~\ref{lemma:0-and-1-phases}
to each individual $1$-phase, hence obtaining
(\ref{equation:target-bound-1-phases}) by the linearity of expectation.

\begin{IntuitionSpotlight}
It remains to establish (\ref{equation:target-bound-0-phases}) which turns out
to be more demanding:
for $0$-phases $\phi$, the upper bound on $\Expect_{v}[\tau(\phi)]$ promised by
\Lem{}~\ref{lemma:0-and-1-phases} includes an additive $w(v)$ term and we have
to make sure that it does not dominate the $\adv{\tau}(\phi)$ and
$\adv{\sigma}(\phi)$ terms too often.
This is done via a classification of the phases with respect to
their starting time.
\end{IntuitionSpotlight}

\paragraph{Early and late phases.}
Recall the definition of
$Y(t) = Y_{1}(t) \xor Y_{2}(t)$
and let $\LateTime$ be the smallest
$t \in [0, \EndTime)$
such that
$\min \{ \int_{t}^{\EndTime} Y(s) d s,  \int_{t}^{\EndTime} \neg Y(s) d s \}
\leq w(v)$.
\NotationLabel{analysis:late-time}
Phase $\phi$ with starting time $t$ is
\LongVersion 
said to be an \emph{early}
phase if
$t < \LateTime$
and a \emph{late} phase if
$t \geq \LateTime$.
\LongVersionEnd 
\ShortVersion 
called \emph{early} if
$t < \LateTime$
and \emph{late} if
$t \geq \LateTime$.
\ShortVersionEnd 
(Intuitively, this means that when an early phase starts, we still have more
than $w(v)$ time units of
$Y(t) = 0$
and more than $w(v)$ time units of
$Y(t) = 1$.)
Let
$P_{\mathrm{early}}$
\NotationLabel{analysis:P-early}
and
$P_{\mathrm{late}}$
\NotationLabel{analysis:P-late}
be the sets of early and late phases, respectively.
Let
$K$
\NotationLabel{analysis:K}
be the number of discontinuity points of $Y(t)$ in the interval
$[0, \LateTime)$.

We would like to take a closer look at the partition of $[0, \EndTime)$ into
phases.
\LongVersion 
To that end, consider
\LongVersionEnd 
\ShortVersion 
Consider
\ShortVersionEnd 
some phase $\phi$ with starting time $T^{-}$ and end
time $T^{+}$.
Fixing
$T^{-} = t$
for some
$t \in [0, \EndTime)$,
the end time $T^{+}$ is a random variable fully determined by the coin tosses
in $\Ancestors(v)$ after time $t$.
\LongVersion 
An important property of this random variable is cast in the following lemma
(together with
two other important properties of the partition of $[0, \EndTime)$ into
phases).
\LongVersionEnd 

\begin{lemma} \label{lemma:partition-time-line-into-phases}
The partition of $[0, \EndTime)$ into phases satisfies the following three
properties: \\
(P1)
if
$t < \LateTime$,
then
$\Expect_{\Ancestors(v)} \left[ \int_{T^{-}}^{T^{+}} X(s) d s \mid T^{-} = t
\right]
\geq
w(v) (1 - 1 / e)$; \\
(P2)
$|P_{\mathrm{early}}| \leq K + 1$; and \\
(P3)
$\Expect_{\Ancestors(v)}[|P_{\mathrm{late}}|] = O (1)$
with an exponentially vanishing upper tail.
\end{lemma}
\LongVersion 
\begin{proof}
We investigate the dynamics of
$(X(t))_{t \in [0, \EndTime)}$
and
$(Y(t))_{t \in [0, \EndTime)}$
(an illustration is provided in \Fig{}~\ref{figure:time-line-app}).
A new request arriving in $\Leaves(v)$ flips $X$ and $Y$.
While $Y$ is affected only by new request arrivals, the dynamic of $X$ is tied
to the actions of $\A$ too.
Specifically, the design of the stilt-walker algorithm ensures that $\A$ can
match on top of $v$ only when $X = 1$ (recall that matching across $v$ does
not affect the partition of $[0, \EndTime)$ to phases).
Suppose that $X = 1$ throughout the infinitesimally small time
interval
$I = [t - d t, t)$;
let $S$ be the stilt in $\Stilts(t')$ to which $v$ belongs for all $t' \in I$
and let
$v' \in \Ancestors(v)$
be the head of $S$.
Then $\A$ matches across $v'$ and on top of $v$ at time $t$ with probability
$\pi(t) = d t / w(v')$
(depending solely on the coin tosses of $v'$), in which case $X$ flips to
$X(t) = 0$.
Since $v'$ is an ancestor of $v$, we know that
$\pi(t) < d t / w(v)$.

We color the time line using the coloring function
$c : [0, \EndTime) \rightarrow \{ 1, 2, \bot \}$
by setting
\[
c(t)
=
\left\{
\begin{array}{ll}
1, & Y(t) = 1 \\
2, & Y(t) = 0
\end{array}
\right.
\]
(note that $\bot$, whose preimage under $c$ is empty, is included in the range
of $c$ for the sake of compatibility with the APP framework).
The key observation now is that the times at which $\A$ matches on top of $v$
can be viewed as the meaningful alternation times of a rate-varying APP
$\Pi_{[0, \EndTime)}$ defined over the time interval $[0, \EndTime)$ with
coloring function $c(\cdot)$ and rate function bounded from above by $1 /
w(v)$ (recall that a rate-varying APP is a generalization of an APP defined in
the end of \Sect{}~\ref{section:alternating-Poisson-process} of the full
version).

\sloppy
Taking $G_{\phi}$ to be the digestion of the iteration in
$\Pi_{[0, \EndTime)}$ that starts at time
$T^{-} = t$,
we notice that
$\int_{T^{-}}^{T^{+}} X(s) d s = G_{\phi}$;
recalling that the definition of $\LateTime$ guarantees that
$\min \left\{ \int_{t}^{\EndTime} \Indicator(c(t) = 1) d t, \int_{t}^{\EndTime}
\Indicator(c(t) = 2) d t \right\} > w(v)$
for every
$t < \LateTime$,
we obtain property (P1) by applying (the rate-varying version of)
\Lem{}~\ref{lemma:APP-low-bound-last-digestion} to $\Pi_{[0, \EndTime)}$.
Property (P2) holds simply by applying
\Lem{}~\ref{lemma:APP-bound-number-meaningful} to the $[0,
\LateTime)$-restriction of $\Pi_{[0, \EndTime)}$.
To obtain property (P3), we consider the $[\LateTime, \EndTime)$-restriction
of $\Pi_{[0, \EndTime)}$, denote its number of meaningful alternation times by
$N$, and observe that $|P_{\mathrm{late}}|$ is stochastically dominated by $N
+ 1$;
the property then follows by \Lem{}~\ref{lemma:APP-bound-number-meaningful}
since
$\min \left\{ \int_{\LateTime}^{\EndTime} \Indicator(c(t) = 1) d t,
\int_{\LateTime}^{\EndTime} \Indicator(c(t) = 2) d t \right\} \leq w(v)$.
\end{proof}
\par\fussy
\LongVersionEnd 

\def\FigureTimeLineApp{
\LongVersion
\begin{figure}
\LongVersionEnd
\ShortVersion
\begin{figure}[h]
\ShortVersionEnd
\begin{center}
\includegraphics[width=\textwidth]{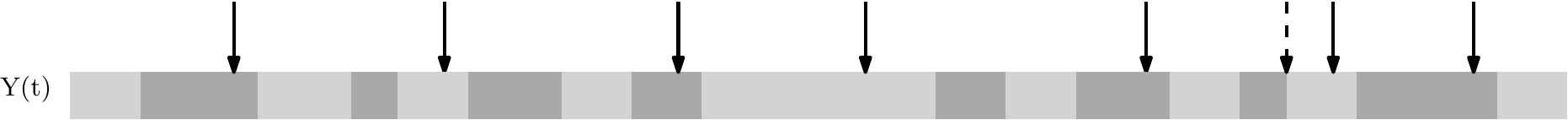}
\end{center}
\caption{ \label{figure:time-line-app}
Interval $[0, \EndTime)$ with time progressing from left to right.
The dark gray and light gray intervals represent the times $t$ at which
$Y(t) = 1$
and
$Y(t) = 0$,
respectively.
The solid vertical arrows represent the times at which $\A$ matches on top of
$v$.
The dashed vertical arrow represent time $\LateTime$.
}
\end{figure}
}
\LongVersion 
\FigureTimeLineApp{}
\LongVersionEnd 

\begin{corollary} \label{corollary:low-bound-early-phase-opt}
If
$\phi$ is a $0$-phase that starts at time
$T^{-} = t < \LateTime$,
then
$\Expect_{\Ancestors(v)}[\adv{\tau}(\phi) + \adv{\sigma}(\phi) \mid T^{-} = t]
\geq
\Omega (w(v))$.
\end{corollary}
\begin{proof}
As $T^{+}$ is a random variable fully determined by the coin tosses in
$\Ancestors(v)$ after time $t$, $\adv{\tau}(\phi)$ and $\adv{\sigma}(\phi)$
are also random variables fully determined by the coin tosses in
$\Ancestors(v)$ after time $t$.
Since
$\adv{\tau}(\phi)
\geq
\int_{t}^{T^{+}} \adv{X}(s) d s$
and since $\phi$ starts with a $0$-subphase $\varphi$ during which
$X = 1$
implies
$\adv{X} = 1$,
the assertion follows from
\Lem{}~\ref{lemma:partition-time-line-into-phases}(P1), recalling that if
$\phi$ contains any subphase other than $\varphi$ (in particular, a
$1$-subphase during which
$X = 1$
does not imply
$\adv{X} = 1$),
then
$\adv{\sigma}(\phi) \geq w(v)$.
\end{proof}

We are now ready to establish (\ref{equation:target-bound-0-phases}).
This is done by defining
$P^{0}_{\mathrm{early}} = P^{0} \cap P_{\mathrm{early}}$
and
$P^{0}_{\mathrm{late}} = P^{0} \cap P_{\mathrm{late}}$
to be the sets of early and late $0$-phases, respectively, and proving the
following two lemmas.

\begin{lemma} \label{lemma:bound-early-phases}
$\Expect_{v, \Ancestors(v)} \left[ \tau(P^{0}_{\mathrm{early}}) \right]
=
O \left( \adv{\tau}(0, \EndTime) + \adv{\sigma}(0, \EndTime) \right)$.
\end{lemma}
\LongVersion 
\begin{proof}
\Lem{}~\ref{lemma:partition-time-line-into-phases}(P2) ensures
that
$|P^{0}_{\mathrm{early}}| \leq K + 1$.
Let
$\phi_{1}, \dots, \phi_{K + 1}$
be the sequence of early $0$-phases, where, for the sake of the analysis, we
introduce a suffix of empty \emph{dummy} phases so that each dummy phase
$\phi_{j}$, $|P^{0}_{\mathrm{early}}| + 1 \leq j \leq K + 1$,
starts and ends at some arbitrary dummy time
$\hat{t} > \EndTime$,
thus ensuring that
$\tau(\phi_{j}) = \adv{\tau}(\phi_{j}) = \adv{\sigma}(\phi_{j}) = 0$.

Fix some
$1 \leq j \leq K + 1$
and let $T^{-}$ and $T^{+}$ be the random variables that capture the starting
time and end time of $\phi_{j}$.
We argue that
\begin{equation} \label{equation:bound-single-phase-conditional}
\Expect_{v, \Ancestors(v)} \left[ \tau(\phi_{j}) | T^{-} = t \right]
\, \leq \,
O \left( \Expect_{\Ancestors(v)} \left[
\adv{\tau}(\phi_{j}) + \adv{\sigma}(\phi_{j}) \mid T^{-} = t
\right] \right)
\end{equation}
for any $t$ in the support of $T^{-}$.
This clearly holds if $\phi_{j}$ is an empty dummy phase (which means that
$t = \hat{t}$),
so assume that
$t < \LateTime$.
Consider the random variable
$\Expect_{v}[\tau(\phi_{j}) | T^{-} = t, T^{+}]$
that maps the event
$T^{+} = s$
(defined over the coin tosses in $\Ancestors(v)$)
to
$\Expect_{v}[\tau(\phi_{j}) | T^{-} = t, T^{+} = s]$.
By \Lem{}~\ref{lemma:0-and-1-phases}, the latter satisfies
$\Expect_{v}[\tau(\phi_{j}) | T^{-} = t, T^{+} = s]
\leq
\adv{\tau}(t, s) + 2 \adv{\sigma}(t, s) + w(v)$.
Therefore,
\begin{align*}
\Expect_{\Ancestors(v)} \left[
\Expect_{v} \left[ \tau(\phi_{j}) | T^{-} = t, T^{+} \right]
\right]
\, \leq \, &
\Expect_{\Ancestors(v)} \left[
\adv{\tau}(\phi_{j}) + 2 \adv{\sigma}(\phi_{j}) \mid T^{-} = t
\right] + w(v) \\
\leq \, &
O \left( \Expect_{\Ancestors(v)} \left[
\adv{\tau}(\phi_{j}) + \adv{\sigma}(\phi_{j}) \mid T^{-} = t
\right] \right) \, ,
\end{align*}
where the last transition follows from
\Cor{}~\ref{corollary:low-bound-early-phase-opt},
thus establishing (\ref{equation:bound-single-phase-conditional}) by the law
of total expectation.

Consider the random variable
$\Expect_{v, \Ancestors(v)}[\tau(\phi_{j}) | T^{-}]$
that maps the event
$T^{-} = t$
(defined over the coin tosses in $\Ancestors(v)$)
to
$\Expect_{v, \Ancestors(v)}[\tau(\phi_{j}) | T^{-} = t]$.
Using the bound provided for the latter by
(\ref{equation:bound-single-phase-conditional}) and applying the law of total
expectation, we conclude that
\[
\Expect_{v, \Ancestors(v)} \left[ \tau(\phi_{j}) \right]
\, = \,
\Expect_{\Ancestors(v)} \left[
\Expect_{v, \Ancestors(v)} \left[ \tau(\phi_{j}) | T^{-} \right]
\right]
\, \leq \,
O \left( \Expect_{\Ancestors(v)} \left[ 
\adv{\tau}(\phi_{j}) + \adv{\sigma}(\phi_{j})
\right] \right) \, .
\]
Therefore, by the linearity of expectation, we derive
\begin{align*}
\Expect_{v, \Ancestors(v)} \left[ \tau(P^{0}_{\mathrm{early}}) \right]
\, = \, &
\sum_{j = 1}^{K + 1} \Expect_{v, \Ancestors(v)} \left[ \tau(\phi_{j}) \right] \\
\leq \, &
\sum_{j = 1}^{K + 1}
O \left( \Expect_{\Ancestors(v)} \left[
\adv{\tau}(\phi_{j}) + \adv{\sigma}(\phi_{j})
\right] \right)
\, = \,
O \left( \Expect_{\Ancestors(v)} \left[
\adv{\tau}(P^{0}_{\mathrm{early}}) + \adv{\sigma}(P^{0}_{\mathrm{early}})
\right] \right)
\end{align*}
which establishes the assertion.
\end{proof}
\LongVersionEnd 
\ShortVersion 
\sloppy
\begin{proof}[Proof (sketch).]
Consider an early $0$-phase $\phi$ and let $T^{-}$ and $T^{+}$ be the random
variables that capture its starting time and end time, respectively.
By \Lem{}~\ref{lemma:0-and-1-phases}, we know that
$\Expect_{v}[\tau(\phi) | T^{-} = t, T^{+} = s]
\leq
\adv{\tau}(t, s) + 2 \adv{\sigma}(t, s) + w(v)$.
Employing \Cor{}~\ref{corollary:low-bound-early-phase-opt}, we can
show that
$\Expect_{v, \Ancestors(v)}[\tau(\phi) | T^{-} = t]
\leq
O (\Expect_{\Ancestors(v)}[\adv{\tau}(\phi) + \adv{\sigma}(\phi) \mid T^{-} =
t])$,
hence, applying the law of total expectation, we prove that
$\Expect_{v, \Ancestors(v)}[\tau(\phi)]
\leq
O (\Expect_{\Ancestors(v)}[\adv{\tau}(\phi) + \adv{\sigma}(\phi)])$.
Then, we employ \Lem{}~\ref{lemma:partition-time-line-into-phases}(P2) to
bound the number of early $0$-phases, so we can use the linearity of
expectation to establish the assertion.
\end{proof}
\par\fussy
\ShortVersionEnd 

\begin{lemma} \label{lemma:bound-late-phases}
$\Expect_{v, \Ancestors(v)} \left[ \tau(P^{0}_{\mathrm{late}}) \right]
\leq
O \left( \adv{\tau}(0, \EndTime) + \adv{\sigma}(0, \EndTime) + w(v) \right)$.
\end{lemma}
\LongVersion 
\begin{proof}
Conditioned on $|P^{0}_{\mathrm{late}}| = m$,
\Lem{}~\ref{lemma:0-and-1-phases} guarantees that
\[
\Expect_{v} \left[ \tau(P^{0}_{\mathrm{late}}) \right]
\, \leq \,
\adv{\tau}(P^{0}_{\mathrm{late}}) + 2 \adv{\sigma}(P^{0}_{\mathrm{late}}) + m
\cdot w(v)
\, \leq \,
\adv{\tau}(P_{\mathrm{late}}) + 2 \adv{\sigma}(P_{\mathrm{late}}) + m \cdot
w(v) \, .
\]
By \Lem{}~\ref{lemma:partition-time-line-into-phases}(P3),
\[
\Expect_{\Ancestors(v)}[|P^{0}_{\mathrm{late}}|]
\, \leq \,
\Expect_{\Ancestors(v)}[|P_{\mathrm{late}}|]
\, \leq \,
O (1)
\]
with an exponentially vanishing upper tail, thus
\[
\Expect_{v, \Ancestors(v)} \left[ \tau(P^{0}_{\mathrm{late}}) \right]
\, \leq \,
O \left(
\adv{\tau}(P_{\mathrm{late}}) + \adv{\sigma}(P_{\mathrm{late}}) + w(v) \right)
\]
which establishes the assertion.
\end{proof}
\LongVersionEnd 
\ShortVersion 
\begin{proof}[Proof (sketch).]
\Lem{}~\ref{lemma:0-and-1-phases} guarantees that 
$\Expect_{v}[\tau(\phi)]
\leq
O (\adv{\tau}(\phi) + \adv{\sigma}(\phi) + w(v))$
for each phase $\phi \in P^{0}_{\mathrm{late}}$.
The assertion follows by
\Lem{}~\ref{lemma:partition-time-line-into-phases}(P3) that practically
provides a constant bound on
$|P^{0}_{\mathrm{late}}| \leq |P_{\mathrm{late}}|$.
\end{proof}
\ShortVersionEnd 

\LongVersion 
\subsection{Lifting the end-of-input signal assumption}
\label{section:lift-end-of-input-assumption}
We now turn to lift the end-of-input signal assumption, showing that
\Thm{}~\ref{theorem:main} holds also without it.
Recall that
$\EndTime = \max \{ \aTime(\rho) \mid \rho \in R \}$
denotes the arrival time of the last request in $R$ and let
$C = \Active(\EndTime)$
and
$F = \Effective(\EndTime)$
be the set of remaining active requests and the set of effective vertices at
time $\EndTime$, respectively.
The analysis presented in \Sect{}~\ref{section:analysis-heart} relies on the
assumption that upon receiving the end-of-input signal at time $\EndTime$, the
algorithm immediately clears all the requests in $C$ by matching across every
vertex in $F$ which contributes
$\sEndCost = \sum_{v \in F} w(v)$
to the space cost of $\A$ (this contribution to the space cost of $\A$ is
taken into account in \Sect{}~\ref{section:analysis-heart}).

An examination of the matching policy of the stilt-walker algorithm reveals
that in reality, the requests in $C$ are indeed cleared by matching across the
vertices in $F$, only that these matching operations are not performed
immediately at time $\EndTime$, but rather at slightly later (random) times,
thus introducing an additional contribution to the time cost of $\A$.
Specifically, taking
$\rho, \rho' \in C$
to be the supporting requests of some effective vertex
$v \in F$,
notice that on expectation, $\A$ matches across $v$ at time
$\EndTime + w(v)$
which accounts for an additional contribution of a $2 w(v)$ term to the
algorithm's expected time cost.
Summing over all vertices in $F$, we conclude that by adding
\[
2 \sum_{v \in F} w(v)
\, < \,
2 \sum_{v \in T} w(v)
\]
to the $\beta$ term in \Thm{}~\ref{theorem:main}, we can lift the end-of-input
assumption as promised.
\LongVersionEnd 

\LongVersion 
\section{A fixed penalty for clearing requests}
\label{section:fixed-penalty}
In this section, we consider the online \emph{MPMDfp} problem:
a variant of MPMD in which the algorithm is allowed to clear any request
$\rho \in R$
at time
$t \geq t(\rho)$
without matching it to another request, incurring a fixed penalty
$p > 0$
(a parameter of the problem), on top of the time cost
$t - t(\rho)$
of $\rho$,
that adds to its total cost.
Notice that in contrast to MPMD, the MPMDfp problem is well defined also for
odd values of $|R|$.

\begin{theorem}
There exists a randomized online MPMDfp algorithm for $\mathcal{M}$ whose
competitive ratio is
$O \left( \log^{2} n + \log \Delta \right)$,
where $n$ is the number of points in the underlying metric space and $\Delta$
is its aspect ratio.
\end{theorem}
\begin{proof}
Consider the underlying $n$-point metric space
$\mathcal{M} = (V, \delta)$
and let
$d = \min_{x \neq y \in V} \delta(x, y)$
and
$D = \max_{x \neq y \in V} \delta(x, y)$
be the minimum and maximum distances between any two distinct points in
$\mathcal{M}$, respectively, so that the aspect ratio of $\mathcal{M}$ is
$\Delta = D / d$.
Assume for the time being that the penalty $p$ satisfies
$d / 2 < p < 2 D$.

Let
$\widehat{\mathcal{M}} = (V \times \{ 1, 2 \}, \widehat{\delta})$
be the metric space defined by setting
\[
\widehat{\delta}((x, i_{x}), (y, i_{y}))
=
\delta(x, y) + p \cdot |i_{x} - i_{y}|
\]
for every
$x, y \in V$
and
$i_{x}, i_{y} \in \{ 1, 2 \}$.
The assumption that
$d / 2 < p < 2 D$
implies that the aspect ratio of $\widehat{\mathcal{M}}$ is proportional to
$\Delta$.
Let
$\widehat{R} = \{ \rho_{1}, \rho_{2} \mid \rho \in R \}$,
where $\rho_{i}$, $i \in \{ 1, 2 \}$, is defined by setting
$\aTime(\rho_{i}) = \aTime(\rho)$
and
$\Location(\rho_{i}) = (\Location(\rho), i)$.

We construct an online MPMDfp algorithm $\A_{\text{fp}}$ with the desired
competitive ratio from the stilt-walker algorithm $\A$ as follows.
Algorithm $\A_{\text{fp}}$ simulates $\A$ on
$\widehat{\mathcal{M}}, \widehat{R}$
and handles the requests in $R$ according to the actions of $\A$ on the
requests in $\widehat{R}$.
Specifically, for every
$\rho \in R$,
if $\A$ matches $\rho_{1}$ to some request $\rho'_{1}$, located in
$V \times \{ 1 \}$,
at time $t$, then $\A_{\text{fp}}$ matches $\rho$ to $\rho'$ at time $t$;
if $\A$ matches $\rho_{1}$ to some request $\rho'_{2}$, located in
$V \times \{ 2 \}$, at time $t$, then $\A_{\text{fp}}$ clears $\rho$ without
matching it (paying the fixed $p$-penalty) at time $t$.

The design of $\A_{\text{fp}}$ and the fact that
$\widehat{\delta}((x, 1), (y, 2)) \geq p$
for every $x, y \in V$
guarantee that
\begin{equation} \label{equation:fp-online}
\Cost_{\A_{\text{fp}}}(R, \mathcal{M}) \leq \Cost_{\A}(\widehat{R},
\widehat{\mathcal{M}}) \, .
\end{equation}
Moreover, the construction of $\widehat{\mathcal{M}}$ and $\widehat{R}$
ensures that if $\adv{\A}$ is an optimal offline MPMD algorithm and
$\adv{\A}_{\text{fp}}$ is an optimal offline MPMDfp algorithm, then
\begin{equation} \label{equation:fp-benchmark}
\Cost_{\adv{\A}}(\widehat{R}, \widehat{\mathcal{M}})
\leq
2 \cdot \Cost_{\adv{\A}_{\text{fp}}}(R, \mathcal{M})
\end{equation}
since $\adv{\A}$ can project the actions of $\adv{\A}_{\text{fp}}$ on each
side of $\widehat{\mathcal{M}}$, matching $\rho_{1}$ to $\rho_{2}$ whenever
$\adv{\A}_{\text{fp}}$ clears $\rho$ without matching it.
The assertion follows since the stilt-walker algorithm $\A$ is $O (\log^{2} n
+ \log \Delta)$-competitive for the MPMD problem.

Now, if
$p < d / 2$,
then an MPMDfp (online or offline) algorithm is always better off clearing the
requests by paying the fixed penalty than by matching them.
Therefore, in this case, the MPMDfp problem over $\mathcal{M}$ can be
decomposed into $n$ independent instances of the MPMDfp over a $1$-point
metric space.
Each such instance (essentially a repeated version of the ski rental problem)
admits an $O (1)$-competitive online algorithm, thus so does the whole
problem.

It remains to consider the case where
$p > 2 D$.
In this case, we construct the metric space $\widehat{\mathcal{M}}$ slightly
differently:
first employ \Thm{}~\ref{theorem:HSBT} to probabilistically embed
$\mathcal{M}$ in a $(1 + \Omega (1 / \log n))$-HSBT $(T, w)$;
then, take two copies of $T$, call them $T_{1}$ and $T_{2}$, and connect them
so that their roots become the children of a new root $\widehat{r}$, extending
the weight function $w$ by setting
$w(\widehat{r}) = p$.
Notice that the resulting metric space is also a $(1 + \Omega (1 / \log
n))$-HSBT whose point set can be renamed
$V \times \{ 1, 2 \}$
so that $(x, i)$ is a leaf of $T_{i}$ for every
$x \in V$
and
$i \in \{ 1, 2 \}$.
The rest of the construction of $\A_{\text{fp}}$ is unchanged.

Although the aspect ratio of the metric space $\widehat{\mathcal{M}}$ in this
case may be large (as large as
$p / d$),
notice that the height of the underlying $(1 + \Omega (1 / \log n))$-HSBT is
still $O (\log \Delta + \log n)$,
where $\Delta$ is the aspect ratio of $\mathcal{M}$.
This establishes the assertion by recalling that the $\log \Delta$ term in the
competitive ratio of the stilt-walker algorithm comes from an upper bound on
the height of its HSBT.
\end{proof}
\LongVersionEnd 

\LongVersion 
\section{The deterministic version of the stilt-walker algorithm}
\label{section:specific-lower-bound}
In this section, we consider the deterministic version of the stilt-walker
algorithm, denoted $\A_{d}$, obtained by replacing the $(1 / w(v))$-rate
exponential timer maintained at each internal vertex $v \in T - \Leaves$ with
a deterministic $w(v)$-timer.
In other words, the matching policy of $\A_{d}$ is similar to that of $\A$
with one difference:
If the last time $\A_{d}$ matched across $v$ was at time $t_{0}$
(take $t_{0} = 0$ if $\A_{d}$ still has not matched across $v$), then the next
time it matches across $v$ is the minimum $t_{1}$ that satisfies
\[
\int_{t_{0}}^{t_{1}} \Indicator(v \in \Effective(t)) \, d t
=
w(v) \, .
\]

\begin{theorem} \label{theorem:specific-lower-bound}
The competitive ratio of $\A_{d}$ on $n$-point
$(1 + \Omega (1 / \log n))$-HSBTs
is $\Omega (n)$.
\end{theorem}
\begin{proof}
Let $n$ be some large power of $2$ and let $T$ be an $n$-leaf perfect binary
tree (with all leaves at depth $\lg n$).
Let
$w : T \rightarrow \Reals_{\geq 0}$
be the weight function defined by setting
$w(x) = 0$
for every leaf $x$; and
$w(v) = (1 + 1 / \lg n)^{\lg(n) - 1 - i}$
for every internal vertex $v$ of depth $i$.
Consider the HSBT
$\mathcal{T} = (T, w)$
and notice that the distance between any two distinct points in $\mathcal{T}$
is $\Theta (1)$.
We name some of the internal vertices and subtrees of $T$ according to the
labels in \Fig{}~\ref{figure:specific-lower-bound}.

Let $\adv{\A}_{d}$ denote the benchmark offline algorithm and take $\epsilon$
to be a small positive real.
For every subtree $T_{j}$, $j = 1, \dots, 6$, fix some arbitrary leaf $x_{j}$
and consider the following scenario $\Gamma$ (refer to
\Fig{}~\ref{figure:specific-lower-bound} for an illustration):

\begin{DenseItemize}

\item
$2$ requests arrive at time $0$ at leaves $x_{1}$ and $x_{6}$ (one each).
$\adv{\A}_{d}$ immediately matches these requests.
Following that, the sole effective vertex of $\A_{d}$ is $v_{1}$ with supporting
leaves $x_{1}$ and $x_{6}$.

\item
$4$ requests arrive at time
$w(v_{1}) - \epsilon$
at leaves
$x_{2}$, $x_{3}$, $x_{4}$, and $x_{5}$ (one each).
Following that, the effective vertices of $\A_{d}$ are: \\
$v_{1}$ with supporting leaves $x_{3}$ and $x_{4}$; \\
$v_{3}$ with supporting leaves $x_{1}$ and $x_{2}$; and \\
$v_{5}$ with supporting leaves $x_{5}$ and $x_{6}$.

\item
The timer of $v_{1}$ expires at time $w(v_{1})$ and $\A_{d}$ matches (across
$v_{1}$) the active requests hosted at leaves $x_{3}$ and $x_{4}$.

\item
$4$ requests arrive at time
$w(v_{1}) + \epsilon$
at leaves
$x_{2}$, $x_{3}$, $x_{4}$, and $x_{5}$ (one each).
Both $\A_{d}$ and $\adv{\A}_{d}$ immediately match the request pairs hosted at
$x_{2}$ and $x_{5}$;
$\adv{\A}_{d}$ also immediately matches the request pairs hosted at $x_{3}$
and $x_{4}$.
Following that, the effective vertices of $\A_{d}$ are: \\
$v_{2}$ with supporting leaves $x_{1}$ and $x_{3}$; and \\
$v_{4}$ with supporting leaves $x_{4}$ and $x_{6}$.

\end{DenseItemize}

Consider the subscenario $\Gamma'$ induced on $\Gamma$ by the time interval
$(0, w(v_{1}) + \epsilon]$.
The key observation is that at the beginning of $\Gamma'$, $\A_{d}$ had
$2$ active requests located at leaves whose LCA is the the root of $T$
($v_{1}$), whereas at its end, $\A_{d}$ has $4$ active requests located at leaves
whose LCAs are the two depth $1$ vertices ($v_{2}$ and $v_{4}$).
On the other hand, $\adv{\A}_{d}$ started and ended subscenario $\Gamma'$ with
no active requests, paying a total cost of $O (\epsilon)$ during that time
period.

Subscenarios analogous to $\Gamma'$ are now applied in a recursive manner to
the subtrees rooted at the depth $i$ vertices of $T$, for
$i = 1, \dots, \lg (n) - 3$.
This results in $\A_{d}$ having active requests at exactly $n / 2$ distinct
leaves of $T$;
to clear all of them, $\A_{d}$ will have to pay $\Omega (n)$ in space cost.
On the other hand, the total cost payed by $\adv{\A}_{d}$ during all these
applications is $O (\epsilon n)$ which can be made arbitrarily small.
Adding the $O (1)$ space cost payed by $\adv{\A}_{d}$ at time $0$ for matching
the first two requests (across $v_{1}$), we conclude that the competitive
ratio of $\A_{d}$ is $\Omega (n)$, as promised.
\end{proof}

\def\FigureSpecificLowerBound{
\begin{figure}
\begin{center}
\includegraphics[width=\textwidth]{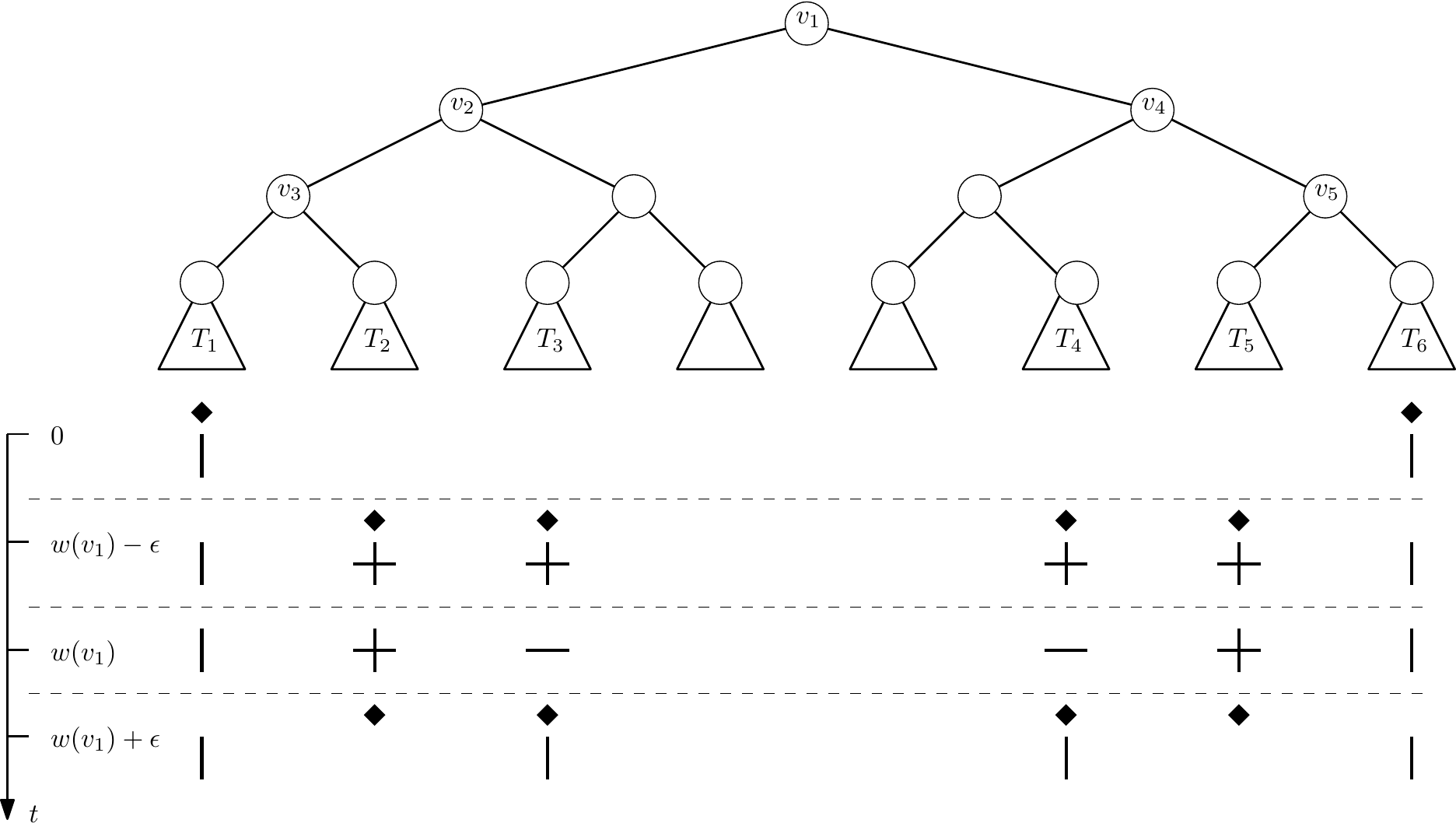}
\end{center}
\caption{ \label{figure:specific-lower-bound}
The perfect binary tree $T$ and scenario $\Gamma$ at times of interest
featured on the left.
For every $j = 1, \dots, 6$ and for every time $t$,
a diamond shape depicts a request arriving at leaf $x_{j}$ at time $t$;
a vertical segment depicts an active request under $\A_{d}$ at leaf $x_{j}$ at
time $t$; and
a horizontal segment depicts an active request under $\adv{\A}_{d}$ at leaf
$x_{j}$ at time $t$.
}
\end{figure}
}
\FigureSpecificLowerBound{}

\LongVersionEnd 

\clearpage
\renewcommand{\thepage}{}

\bibliographystyle{abbrv}
\bibliography{references}

\LongVersion 
\clearpage
\pagenumbering{roman}
\appendix

\renewcommand{\theequation}{A-\arabic{equation}}
\setcounter{equation}{0}

\begin{center}
\textbf{\large{APPENDIX}}
\end{center}

\section{Probabilistic embedding of arbitrary metric spaces in HSBTs}
\label{appendix:embedding-in-HSBT}
Our goal in this section is to prove \Thm{}~\ref{theorem:HSBT}.
The main ingredient in this proof is the following celebrated theorem of
Fakcharoenphol et al.~\cite{FakcharoenpholRT2004}.

\begin{theorem}[\cite{FakcharoenpholRT2004}] \label{theorem:FRT}
Consider some $n$-point metric space
$(V, \delta)$
and let $\mathcal{U}$ be the set of all $2$-HSTs over $V$ with distance
functions $\delta_{\mathcal{T}}$ that dominate $\delta$ in the sense that
$\delta_{\mathcal{T}}(x, y) \geq \delta(x, y)$
for every $x, y \in V$.
There exists a probability distribution $\mathcal{P}$ over $\mathcal{U}$ such
that
$\Expect_{(V, \delta_{\mathcal{T}}) \in \mathcal{P}}[\delta_{\mathcal{T}}(x, y)]
\leq
O (\log n) \cdot \delta(x, y)$
for every $x, y \in V$.
Moreover, the probability distribution $\mathcal{P}$ can be sampled
efficiently.
\end{theorem}

Observe that by the definition of HSTs, the rooted trees realizing the
$2$-HSTs promised by \Thm{}~\ref{theorem:FRT} are of height
$O (\log \Delta)$,
where
$\Delta = \frac{\max_{x, y \in V} \delta(x, y)}{\min_{x, y \in V} \delta(x, y)}$
is the aspect ratio of the metric space
$(V, \delta)$.
These rooted trees have arbitrary degrees, whereas \Thm{}~\ref{theorem:HSBT}
requires rooted trees with degrees at most $2$.
We resolve this obstacle with the help of the following lemma, proved by
Patt-Shamir \cite{PattShamir2015}.

\begin{lemma}[\cite{PattShamir2015}] \label{lemma:Boaz}
Consider some $n$-leaf rooted tree $T$.
There exist
a (rooted) full binary tree $T'$ and an injection
$f : T \rightarrow T'$
such that \\
(1)
$f(v)$ is an ancestor of $f(u)$ in $T'$ if and only if $v$ is an ancestor of
$u$ in $T$; \\
(2)
$\Depth_{T'}(f(v)) \leq \Depth_{T'}(f(p^{T}(v))) + O (\log n)$, where
$\Depth_{T'}(\cdot)$ denotes the depth operator in tree $T'$; and \\
(3)
$\Height(T') = O (\Height(T) + \log n)$.
\end{lemma}

Consider some $n$-point tree metric space $(T, w)$ in the support of the
probability distribution promised by \Thm{}~\ref{theorem:FRT} and let $T'$ and
$f : T \rightarrow T'$
be the full binary tree and injection obtained by applying
\Lem{}~\ref{lemma:Boaz} to $T$.
We construct a weight function
$w' : T' \rightarrow \Reals_{\geq 0}$
on the vertices of $T'$ by first setting
$w'(v) = 2 \cdot w(f^{-1}(v))$
for every
$v \in f(T)$,
and then fixing
$w'(v) = w'(p^{T'}(v)) / (1 + \Omega (1 / \log n))$
for every
$v \notin f(T)$.
\Lem{}~\ref{lemma:Boaz} guarantees that
$(T', w')$ is a $(1 + \Omega (1 / \log n))$-HSBT.
Taking $\delta$ and $\delta'$ to be the distance functions of $(T, w)$ and
$(T', w')$, respectively, we observe that
\[
\delta(x, y)
\leq
\delta'(x, y)
\leq
2 \delta(x, y)
\]
for every two points $x, y$ in the metric space(s), thus establishing
\Thm{}~\ref{theorem:HSBT}.

\LongVersionEnd 

\clearpage
\renewcommand{\thepage}{}

\ShortVersion 

\begin{center}
\textbf{\large{FIGURES}}
\end{center}

\PseudocodeStiltWalker{}

\FigureApp{}


\FigureSubphaseApp{}



\ShortVersionEnd 

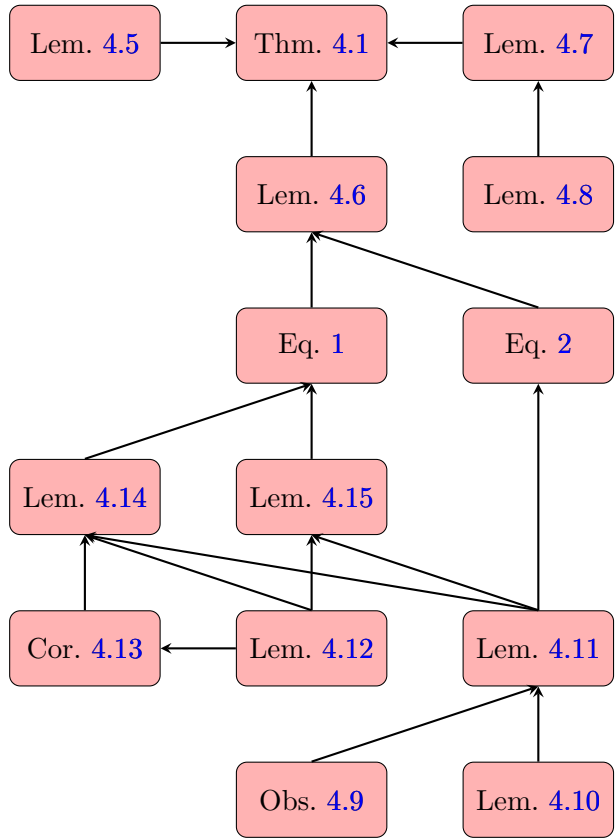
\begin{figure}[h]
\begin{center}
\begin{tikzpicture}[box/.style={draw, rectangle, rounded corners, minimum
width=2cm, minimum height=1cm, text centered, draw=black, fill=red!30},
arrow/.style={thick,->,>=stealth}]

\node (chart-theorem-main)
[box]
{\Thm{}~\ref{theorem:main}};

\node (chart-lemma-expressing-costs-by-potentials)
[box, left=of chart-theorem-main]
{\Lem{}~\ref{lemma:expressing-costs-by-potentials}};

\node (chart-lemma-bounding-time-potential)
[box, below=of chart-theorem-main]
{\Lem{}~\ref{lemma:bounding-time-potential}};

\node (chart-lemma-bounding-space-potential)
[box, right=of chart-theorem-main]
{\Lem{}~\ref{lemma:bounding-space-potential}};

\node (chart-lemma-phase-app)
[box, below=of chart-lemma-bounding-space-potential]
{\Lem{}~\ref{lemma:phase-app}};

\node (chart-equation-target-bound-0-phases)
[box, below=of chart-lemma-bounding-time-potential]
{Eq.~\ref{equation:target-bound-0-phases}};

\node (chart-equation-target-bound-1-phases)
[box, right=of chart-equation-target-bound-0-phases]
{Eq.~\ref{equation:target-bound-1-phases}};

\node (chart-lemma-bound-late-phases)
[box, below=of chart-equation-target-bound-0-phases]
{\Lem{}~\ref{lemma:bound-late-phases}};

\node (chart-lemma-bound-early-phases)
[box, left=of chart-lemma-bound-late-phases]
{\Lem{}~\ref{lemma:bound-early-phases}};

\node (chart-lemma-partition-time-line-into-phases)
[box, below=of chart-lemma-bound-late-phases]
{\Lem{}~\ref{lemma:partition-time-line-into-phases}};

\node (chart-equation-low-bound-early-phase-opt)
[box, left=of chart-lemma-partition-time-line-into-phases]
{\Cor{}~\ref{corollary:low-bound-early-phase-opt}};

\node (chart-lemma-0-and-1-phases)
[box, right=of chart-lemma-partition-time-line-into-phases]
{\Lem{}~\ref{lemma:0-and-1-phases}};

\node (chart-lemma-bound-time-cost-0-subphase)
[box, below=of chart-lemma-0-and-1-phases]
{\Lem{}~\ref{lemma:bound-time-cost-0-subphase}};

\node (chart-observation-bound-time-cost-1-subphase)
[box, left=of chart-lemma-bound-time-cost-0-subphase]
{\Obs{}~\ref{observation:bound-time-cost-1-subphase}};

\draw[arrow]
(chart-lemma-expressing-costs-by-potentials.east) --
(chart-theorem-main.west);

\draw[arrow]
(chart-lemma-bounding-time-potential.north) --
(chart-theorem-main.south);

\draw[arrow]
(chart-lemma-bounding-space-potential.west) --
(chart-theorem-main.east);

\draw[arrow]
(chart-lemma-phase-app.north) --
(chart-lemma-bounding-space-potential.south);

\draw[arrow]
(chart-equation-target-bound-0-phases.north) --
(chart-lemma-bounding-time-potential.south);

\draw[arrow]
(chart-equation-target-bound-1-phases.north) --
(chart-lemma-bounding-time-potential.south);

\draw[arrow]
(chart-lemma-bound-early-phases.north) --
(chart-equation-target-bound-0-phases.south);

\draw[arrow]
(chart-lemma-0-and-1-phases.north) --
(chart-lemma-bound-early-phases.south);

\draw[arrow]
(chart-lemma-bound-late-phases.north) --
(chart-equation-target-bound-0-phases.south);

\draw[arrow]
(chart-equation-low-bound-early-phase-opt.north) --
(chart-lemma-bound-early-phases.south);

\draw[arrow]
(chart-lemma-partition-time-line-into-phases.north) --
(chart-lemma-bound-early-phases.south);

\draw[arrow]
(chart-lemma-partition-time-line-into-phases.west) --
(chart-equation-low-bound-early-phase-opt.east);

\draw[arrow]
(chart-lemma-0-and-1-phases.north) --
(chart-equation-target-bound-1-phases.south);

\draw[arrow]
(chart-observation-bound-time-cost-1-subphase.north) --
(chart-lemma-0-and-1-phases.south);

\draw[arrow]
(chart-lemma-bound-time-cost-0-subphase.north) --
(chart-lemma-0-and-1-phases.south);

\draw[arrow]
(chart-lemma-0-and-1-phases.north) --
(chart-lemma-bound-late-phases.south);

\draw[arrow]
(chart-lemma-partition-time-line-into-phases.north) --
(chart-lemma-bound-late-phases.south);

\end{tikzpicture}
\end{center}
\caption{\label{figure:claim-chart}
A schematic overview of the analysis carried out in
\Sect{}~\ref{section:analysis-heart}, depicting the interdependencies between
its components.
An arrow pointing from A to B indicates that the proof corresponding
to B depends on the statement corresponding to A.
The proofs of \Lem{}\ \ref{lemma:phase-app},
\ref{lemma:bound-time-cost-0-subphase}, and
\ref{lemma:partition-time-line-into-phases} are based on the APP machinery
developed in \Sect{}~\ref{section:alternating-Poisson-process}.
}
\end{figure}

\begin{table}[h]
\begin{center}
{\small
\begin{tabular}{|l|l|l|}
\hline
\textsf{Notation} & \textsf{Definition} & \textsf{Defined on page} \\
\hline

$\Location(\rho)$ &
location of $\rho$ &
\NotationPageRef{model:location} \\

$\aTime(\rho)$ &
arrival time of $\rho$ &
\NotationPageRef{model:arrival-time} \\

$\sCost(\rho)$ &
space cost of $\rho$ &
\NotationPageRef{model:space-cost-request} \\

$\tCost(\rho)$ &
time cost of $\rho$ &
\NotationPageRef{model:time-cost-request} \\

$\sCost(R)$ &
space cost of $R$ &
\NotationPageRef{model:space-cost-instance} \\

$\tCost(R)$ &
time cost of $R$ &
\NotationPageRef{model:time-cost-instance} \\

$\Cost(R, \mathcal{M})$ &
total cost &
\NotationPageRef{model:total-cost} \\

$p(v)$ &
parent of $v$ &
\NotationPageRef{tree:parent} \\

$T(v)$ &
subtree rooted at $v$ &
\NotationPageRef{tree:subtree} \\

$\Leaves(v)$ &
leaves of $T(v)$ &
\NotationPageRef{tree:subtree-leaves} \\

$\Ancestors(v)$ &
ancestors of $v$ &
\NotationPageRef{tree:ancestors} \\

$\Depth(v)$ &
depth of $v$ &
\NotationPageRef{tree:depth} \\

$\Height(T)$ &
height of $T$ &
\NotationPageRef{tree:height} \\

$\LCA(x, y)$ &
least common ancestor of $x$ and $y$ &
\NotationPageRef{tree-lca} \\

$\Active_{v}(t)$ &
set of requests with locations in $\Leaves(v)$ active at time $t$ &
\NotationPageRef{alg:active} \\

$\Active(t)$ &
$\Active_{v}(t)$ for $v = r$ &
\NotationPageRef{alg:active-root} \\

$\Odd(t)$ &
set of odd vertices at time $t$ &
\NotationPageRef{alg:odd} \\

$\adv{\Odd}(t)$ &
set of vertices odd under $\adv{\A}$ at time $t$ &
\NotationPageRef{alg:adv-odd} \\

$\Stilts(t)$ &
set of stilts induced by the odd vertices at time $t$ &
\NotationPageRef{alg:stilts} \\

$\Heads(t)$ &
set of heads of the stilts in $\Stilts(t)$ &
\NotationPageRef{alg:heads} \\

$\Effective(t)$ &
set of effective vertices at time $t$ &
\NotationPageRef{alg:effective} \\

$\EndTime$ &
arrival time of the last request &
\NotationPageRef{analysis:end-time} \\

$\sEndCost$ &
space cost of matching the active requests at time $\EndTime$ &
\NotationPageRef{analysis:end-cost} \\

$\tau_{v}(t_{0}, t_{1})$ &
time potential $v$ accumulates during $[t_{0}, t_{1})$ under $\A$ &
\NotationPageRef{analysis:tau} \\

$\adv{\tau}_{v}(t_{0}, t_{1})$ &
time potential $v$ accumulates during $[t_{0}, t_{1})$ under $\adv{\A}$ &
\NotationPageRef{analysis:adv-tau} \\

$\sigma_{v}(t_{0}, t_{1})$ &
space potential $v$ accumulates during $[t_{0}, t_{1})$ under $\A$ &
\NotationPageRef{analysis:sigma} \\

$\adv{\sigma}_{v}(t_{0}, t_{1})$ &
space potential $v$ accumulates during $[t_{0}, t_{1})$ under $\adv{\A}$ &
\NotationPageRef{analysis:adv-sigma} \\

$X_{i}(t)$ &
$\Indicator(u_{i} \in \Odd(t))$ &
\NotationPageRef{analysis:variable-X-i} \\

$\adv{X}_{i}(t)$ &
$\Indicator(u_{i} \in \adv{\Odd}(t))$ &
\NotationPageRef{analysis:variable-adv-X-i} \\

$X(t)$ &
$X_{1}(t) \xor X_{2}(t)$ &
\NotationPageRef{analysis:variable-X} \\

$\adv{X}(t)$ &
$\adv{X}_{1}(t) \xor \adv{X}_{2}(t)$ &
\NotationPageRef{analysis:variable-adv-X} \\

$Y_{i}(t)$ &
$|\{ \rho \in R \mid \Location(\rho) \in \Leaves(u_{i}) \land
\aTime(\rho) \leq t \}| \pmod{2}$ &
\NotationPageRef{analysis:variable-Y-i} \\

$Y(t)$ &
$Y_{1}(t) \xor Y_{2}(t)$ &
\NotationPageRef{analysis:variable-Y} \\

$P^{0}$ &
set of $0$-phases &
\NotationPageRef{analysis:P-0} \\

$P^{1}$ &
set of $1$-phases &
\NotationPageRef{analysis:P-1} \\

$\LateTime$ &
smallest $t$ s.t.\ 
$\min \left\{ \int_{t}^{\EndTime} Y(t) d t,
\int_{t}^{\EndTime} \neg Y(t) d t \right\} \leq w(v)$ &
\NotationPageRef{analysis:late-time} \\

$P_{\mathrm{early}}$ &
set of early phases &
\NotationPageRef{analysis:P-early} \\

$P_{\mathrm{late}}$ &
set of late phases &
\NotationPageRef{analysis:P-late} \\

$K$ &
number of discontinuity points of $Y(t)$ in $[0, \LateTime)$ &
\NotationPageRef{analysis:K} \\

\hline
\end{tabular}
} 
\end{center}
\caption{\label{table:notation}
A table of notations.
}
\end{table}

\end{document}